\documentclass[10pt, twocolumn]{article}

\usepackage[utf8]{inputenc} %
\usepackage[english]{babel}
\usepackage[T1]{fontenc}    %
\usepackage{hyperref}       %
\usepackage{url}            %
\usepackage{booktabs}       %
\usepackage{amsfonts}       %
\usepackage{amsmath}       %
\usepackage{nicefrac}       %
\usepackage{microtype}      %
\usepackage{amsthm}
\usepackage{graphicx}
\usepackage[sorting=none,style=ieee, backend=bibtex]{biblatex}
\usepackage{csquotes}
\usepackage{xcolor}
\usepackage{apptools}
\usepackage{mathtools,cuted}

\usepackage{sectsty}
\sectionfont{\fontsize{12}{15}\selectfont}
\subsectionfont{\fontsize{10}{15}\selectfont}

\newcommand{\mcF}{F_\mathrm{RNN}}

\newcommand*\dif{\mathop{}\!\mathrm{d}}

\newcommand{\mA}{\mathbf{A}}
\newcommand{\mB}{\mathbf{B}}
\newcommand{\mC}{\mathbf{C}}

\newcommand{\mH}{\mathbf{H}}
\newcommand{\mI}{\mathbf{I}}
\newcommand{\mJ}{\mathbf{J}}
\newcommand{\mK}{\mathbf{K}}
\newcommand{\mL}{\mathbf{L}}

\newcommand{\mQ}{\mathbf{Q}}

\newcommand{\mU}{\mathbf{U}}

\newcommand{\mW}{\mathbf{W}}

\newcommand{\va}{\mathbf{a}}
\newcommand{\vb}{\mathbf{b}}

\newcommand{\ve}{\mathbf{e}}
\newcommand{\vf}{\mathbf{f}}
\newcommand{\vg}{\mathbf{g}}
\newcommand{\vh}{\mathbf{h}}

\newcommand{\vk}{\mathbf{k}}

\newcommand{\vs}{\mathbf{s}}

\newcommand{\vu}{\mathbf{u}}
\newcommand{\vv}{\mathbf{v}}

\newcommand{\vx}{\mathbf{x}}
\newcommand{\vy}{\mathbf{y}}
\newcommand{\vz}{\mathbf{z}}
\newcommand{\vdelta}{\boldsymbol{\delta}}

\newcommand{\vone}{\boldsymbol{1}}
\newcommand{\vsigma}{\boldsymbol{\sigma}}

\newcommand{\Lreg}{L_\mathrm{reg}}

\newcommand{\mWin}{\mW^\mathrm{in} }
\newcommand{\tWin}{\mW^{\mathrm{in,}0} }
\newcommand{\mWout}{\mW^\mathrm{out} }
\newcommand{\tWout}{\mW^{\mathrm{out,}0}} 
\newcommand{\mcW}{\mathcal W}
\newcommand{\mcN}{\mathcal N}
\newcommand{\mcO}{\mathcal O}
\newcommand{\tmcW}{\mathcal{W}^0}
\newcommand{\tmJ}{\mJ^0}
\newcommand{\reals}{\mathbb R}
\newcommand{\tvx}{\vx^0}
\newcommand{\tvy}{\vy^0}

\newcommand{\tx}{x^0}

\newcommand{\tJ}{J^0}
\newcommand{\mcK}{\mathcal{K}}
\newcommand{\mcI}{\mathcal{I}}
\newcommand{\mcJ}{\mathcal{J}}
\newcommand{\mcX}{\mathcal{X}}
\newcommand{\tmcX}{\mathcal{X}^0}
\newcommand{\Tr}{\mathrm{Tr}}
\newcommand{\diag}{\mathrm{diag}}

\newcommand{\mcB}{\mathcal{B}}

\newcommand{\bmC}{\bar{\mC}}

\newcommand\extrafootertext[1]{%
    \bgroup
    \renewcommand\thefootnote{\fnsymbol{footnote}}%
    \renewcommand\thempfootnote{\fnsymbol{mpfootnote}}%
    \footnotetext[0]{#1}%
    \egroup
}

\newtheorem{lemma}{Lemma}
\newtheorem{proposition}{Proposition}
\newtheorem*{proposition*}{Proposition}

\AtAppendix{\counterwithin{proposition}{section}}
\AtAppendix{\counterwithin{lemma}{section}}

\title{\bf\large{Synaptic balancing: a biologically plausible local learning rule that provably increases neural network noise robustness without sacrificing task performance}}

\author{%
  Christopher H.~Stock\textsuperscript{1}\thanks{chstock@stanford.edu}, Sarah E. Harvey\textsuperscript{2},
  Samuel A. Ocko\textsuperscript{2}, Surya Ganguli\textsuperscript{2,3}\\
  \small{\textbf{1} Neuroscience Graduate Program, Stanford University School of Medicine, Stanford, CA}\\
  \small{\textbf{2} Department of Applied Physics, Stanford University, Stanford, CA}\\
  \small{\textbf{3} Stanford Institute for Human-Centered Artificial Intelligence, Stanford University, Stanford, CA}\\
}

\date{}

\begin{document}

\renewcommand{\abstractname}{\vspace{-\baselineskip}}  %

\twocolumn[
  \begin{@twocolumnfalse}
    \maketitle
    \begin{abstract}
    We introduce a novel, biologically plausible local learning rule that provably increases the robustness of neural dynamics to noise in nonlinear recurrent neural networks with homogeneous nonlinearities. Our learning rule achieves higher noise robustness without sacrificing performance on the task and without requiring any knowledge of the particular task. The plasticity dynamics---an integrable dynamical system operating on the weights of the network---maintains a multiplicity of conserved quantities, most notably the network's entire temporal map of input to output trajectories. The outcome of our learning rule is a synaptic balancing between the incoming and outgoing synapses of every neuron. This synaptic balancing rule is consistent with many known aspects of experimentally observed heterosynaptic plasticity, and moreover makes new experimentally testable predictions relating plasticity at the incoming and outgoing synapses of individual neurons. Overall, this work provides a novel, practical local learning rule that exactly preserves overall network function and, in doing so, provides new conceptual bridges between the disparate worlds of the neurobiology of heterosynaptic plasticity, the engineering of regularized noise-robust networks, and the mathematics of integrable Lax dynamical systems. 
    \vspace{1.4cm}
    \end{abstract}
  \end{@twocolumnfalse}
  ]

\section{Introduction}

\extrafootertext{* \texttt{chstock@stanford.edu}}  %

As any neural circuit computes, it is subject to additional fluctuations either within the circuit or from other brain regions \cite{faisal2008noise, tolhurst1983statistical, mainen1995reliability, shadlen1998variable}, 
and these fluctuations can impair performance \cite{Rumyantsev2020-nl, Ganguli2008-ad, Ganguli2010-ou, Kadmon2020-ro}.
The fundamental puzzle we address is what kind of plasticity rule can make the dynamics of a  neural circuit more robust to such fluctuations in a manner that: (1) works for any task; (2) is completely agnostic to the learning rule used to solve a task; and (3) does not impair network performance after a task is learned.
Our main contribution is the discovery of a plasticity rule that provably accomplishes all of three objectives for any nonlinear recurrent neural network with a homogeneous nonlinearity, such as the commonly studied rectified linear function.
Furthermore, our plasticity rule is biologically plausible and computable using only information that is locally available at each synapse and its adjacent neurons.
Finally, at a mathematical level our plasticity rule connects to the theory of integrable dynamical systems \cite{lax1968integrals, helmke1994optimization, chu2008linear} and heat diffusion, while experimentally, its features are similar to observed aspects of heterosynaptic plasticity \cite{chistiakova2014heterosynaptic, chistiakova2015homeostatic, oh2015heterosynaptic, el2018locally, field2020heterosynaptic}. 

The key idea behind our learning rule is to exploit a many-to-one mapping between patterns of synaptic strength and the task, or the temporal input-output map implemented by a recurrent neural network.
Indeed, modern neural network models of animal behavior typically possess a large number of tunable parameters---far more than necessary to perform a given simple behavior.
In this overparameterized regime, it has been observed across multiple behaviors and organisms that many distinct model configurations are able to generate equivalent levels of task performance \cite{prinz2004similar, naumann2016whole}.
This observation has spurred theoretical and numerical investigations into specific means by which a task may constrain network connectivity and function \cite{Maheswaranathan2019-dw, gao2015simplicity, barak2013fixed, sussillo2015neural, mante2013context}.
However, the complex, nonlinear nature of many classes of network models in neuroscience---notably recurrent neural networks (RNNs)---has made it difficult in many cases to obtain a precise theoretical characterization of the space of synaptic patterns that all map to the same task \cite{barak2017recurrent}. 

Besides a theoretical interest in formally characterizing equivalence classes of synaptic weights that solve the same task, we are also motivated by a complementary biological question: how might neural systems actually implement local plasticity rules which take advantage of network overparameterization to maintain desirable functional properties, including task performance, in the presence of internal and external sources of variation?
Studies of homeostatic plasticity in cortical circuits have shed light on synaptic mechanisms thought to maintain stable computation, including synaptic scaling, heterosynaptic plasticity, and other compensatory processes \cite{turrigiano2008self, turrigiano2017dialectic, chen2013heterosynaptic, chistiakova2015homeostatic, zenke2017hebbian}. 
However, a fundamental neuroscientific question remains: can such local homeostatic or compensatory plasticity rules operate in such a way so as not to impair task performance, while still accruing some other additional benefit? 

In this work, we provide insights into both the nature of overparameterization in recurrent neural networks and how local plasticity rules might exploit this overparameterization to specifically improve network robustness without impairing task performance.  First we precisely characterize the equivalence classes of synaptic weights that all solve the same task in nonlinear recurrent networks with homogenous nonlinearities.  Using this characterization, we derive a simple associated plasticity rule that locally modifies synaptic weights while remaining within these equivalence classes. Intriguingly, we find our theoretically derived plasticity rule resembles experimentally observed heterosynaptic plasticity rules.

\subsection{Outline of this paper}

The structure of this paper is as follows. 
In \S\ref{sec:TPT} we introduce the nonlinear recurrent network model and a transformation of network weights that exactly preserves task performance.
In \S\ref{sec:robustness} we formalize a notion of noise robustness in recurrent networks and show that this quantity is convex with respect to the coordinates of the task-preserving transformation.
In \S\ref{sec:synaptic-balancing-dynamics} we derive a dynamical transformation of the network, called synaptic balancing, which maximizes robustness while exactly maintaining task performance, and which is implementable entirely by local update rules.
In \S\ref{sec:equilibrium-existence-stability} we show that synaptic balancing is exponentially stable in any recurrent network which does not contain irreducible feedforward structure.
In \S\ref{sec:generalizations-lax} we introduce several generalizations of our model and show that synaptic balancing is an instance of a broader class of integrable dynamical systems on synaptic weight matrices known as Lax dynamical systems.  These dynamical systems lead to isospectral flows on matrices that preserve all eigenvalues of the matrix. 

Turning to the behavior of synaptic balancing alongside task-relevant learning dynamics, in \S\ref{sec:regularized-networks-balanced} we prove that a broad class of regularized networks naturally approaches the equilibrium of our rule through training.
Conversely, in \S\ref{sec:trained-network-improvement} we show empirically that our rule is able to improve the task performance of trained networks in previously unseen noisy regimes.

Finally, we address the role of synaptic balancing as a candidate local plasticity rule for maintaining stable network computation in neural circuits. In \S\ref{sec:exact-approx-solutions} we present exact and approximate solutions to the trajectory of synaptic balancing, deriving a formal connection between synaptic balancing and heat diffusion in a network.
In \S\ref{sec:heterosynaptic-plasticity-perturbation} we draw a connection between synaptic balancing and experimentally observed patterns of heterosynaptic plasticity in cortical synapses.

\section{A task-preserving transformation defines a manifold of equivalent recurrent networks}
\label{sec:TPT}

In this section, we first describe a class of nonlinear recurrent neural networks that has been extensively studied in diverse contexts \cite{Vogels2005-sr,sompolinsky1988chaos, mante2013context,Ganguli2008-ad,Ganguli2010-ou,barak2017recurrent,Maheswaranathan2019-kk,Maheswaranathan2019-dw}. We then describe a natural symmetry acting on the weight space of these neural networks that {\it exactly preserves} the entire temporal mapping of input to output trajectories.
Since the input-output mapping determines the task a neural network performs, the action of this symmetry enables us to traverse the manifold of weight configurations that preserve the task performed by any neural network. 

\subsection{Recurrent network model}
Consider a recurrent rate network of $N$ neurons with an $N \times N$ synaptic weight matrix $\mJ$, such that neuron $j$ is connected to neuron $i$ with synaptic weight $J_{ij}$.
The vector of neural activity $\vx \in \reals^N$ and readout vector $\vy$ evolve under a time-varying external input $\vu$ as
\begin{align}
\label{eq:neural-dynamics}
  \tau \dot \vx &= \vf_\mcW(\vx, \vu) = -\vx + \mJ \phi(\vx) + \mWin \vu,  \\
\label{eq:readout}
  \vy &= \mWout \vx.
\end{align}
Here $\tau$ is a fast timescale of neural dynamics, $\phi$ is a typically nonlinear scalar function applied element-wise to $\vx$, and the tuple $\mcW = (\mWin, \, \mJ, \, \mWout)$ is the {weight configuration} which parameterizes the network.

We assume for simplicity that  neural activity  is initialized at the origin and runs until some time $T$.
The dynamics of the full network, and in particular the output trajectory $\vy(t)$, are a deterministic function of the input trajectory $\vu(t)$ and the weight configuration $\mcW$.
More formally, the input-output map $\mcF$ computed by the network is the map from input to output trajectories under a particular weight configuration:  
\begin{align}
    \label{eq:input-output-map}
    \mcF(\mcW): \: \{\vu(t)\}_{t=0}^T \mapsto \{\vy(t)\}_{t=0}^T
\end{align}

In this work, we introduce a model of synaptic updates whose properties are well behaved when $\phi$ is a homogeneous function; that is, when $\phi(\alpha x) = \alpha \phi(x)$ for all $x\in\reals$ and $\alpha\in\reals^+$.
Common activation functions which satisfy this condition include the linear ($\phi(x)=x$) and rectified linear ($\phi(x) = \max(0,x)$) units.
 The rest of this work assumes that $\phi$ is homogeneous.

\subsection{Task-preserving transformation}

\begin{figure*}
\makebox[\textwidth][c]{
\includegraphics[width=6in]{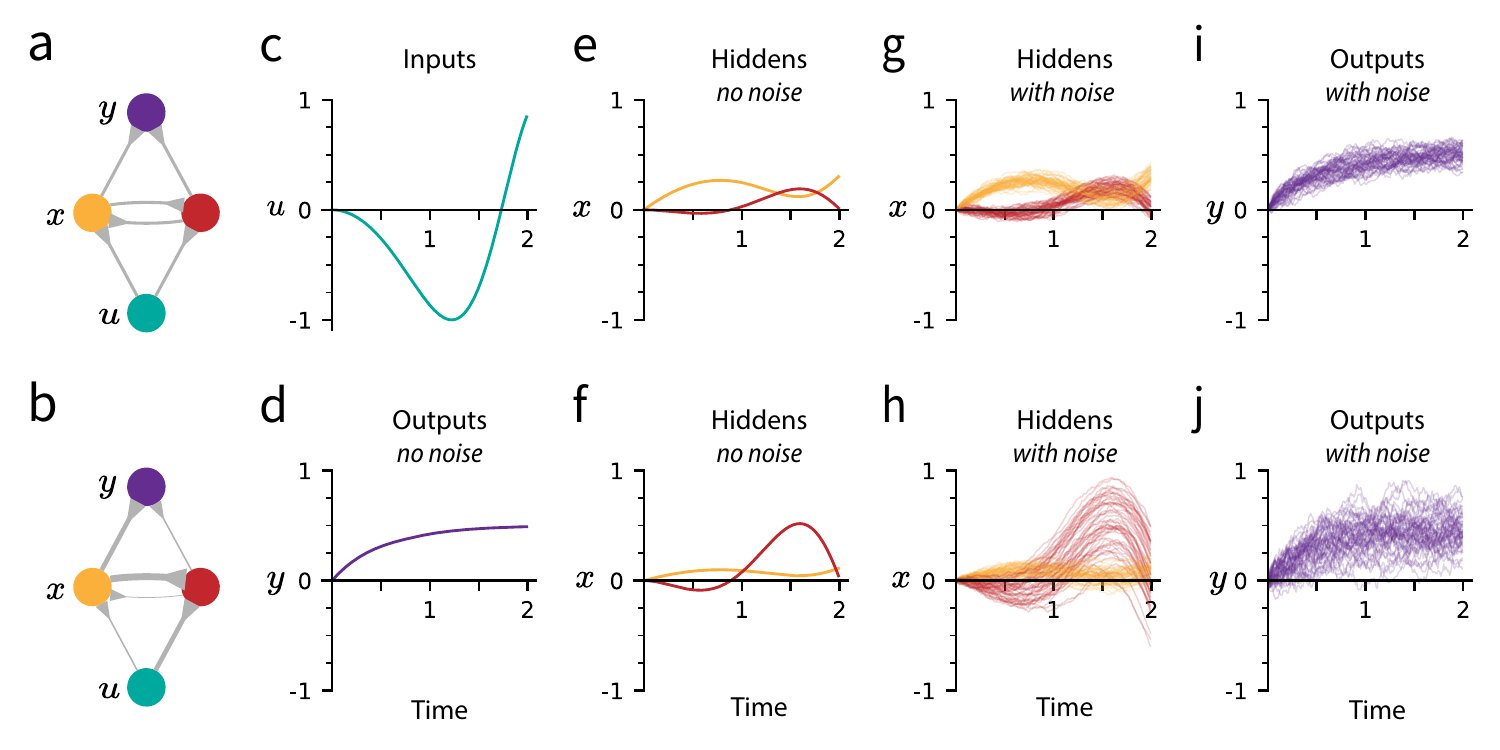}
}
\caption{ 
Two nonlinear recurrent networks consisting of input, hidden, and output neurons possess an identical input-output relationship, but behave differently in the presence of noise.
\textbf{(a-b)} A single input $u$ (teal) drives two recurrent neurons $x_1$ (yellow) and $x_2$ (red) which project to a single output $y$ (purple).
The connectivity patterns of the two networks are related by the task-preserving transformation \eqref{eq:task-preserving-transformation}.
Line thickness denotes synaptic strength.
\textbf{(c)} Input trajectory, fed to both networks. Horizontal axis in all panels is time.
\textbf{(d)} Output trajectory, produced by both networks when run according to the deterministic dynamics \eqref{eq:neural-dynamics}.
\textbf{(e-f)} Hidden unit neural activity under deterministic dynamics, for networks \textbf{a} and \textbf{b} respectively. Trajectories are identical up to a per-neuron scale factor, determined by the parameters of the task-preserving transformation.
\textbf{(g-h)} Hidden neuron neural activity for networks \textbf{a} and \textbf{b} respectively  when additive Gaussian noise is injected into the neural dynamics \eqref{eq:neural-dynamics}.  
50 trials are shown.
\textbf{(i-j)} Output neuron neural activity in the noisy case. 
}
\label{fig:task-preserving-transformation}
\end{figure*}

We now describe a symmetry in weight space that exactly preserves the map \eqref{eq:input-output-map}. Indeed, the observation that many distinct weight configurations may yield quantitatively similar network behavior is widespread both in neuroscience, where simulations of the crustacean stomatogastric ganglion found a range of synaptic strengths producing a given network output \cite{prinz2004similar}, and in machine learning, in which non-convex cost functions possess many roughly equivalent local minima \cite{dauphin2014identifying,Maheswaranathan2019-dw}.
In the class of networks we study, part of the degeneracy between weights and a given task originates from a symmetry in weight space that exactly preserves the deterministic input-output map computed by the network.
Intuitively, this symmetry scales neurons' inputs while reciprocally scaling their outputs, such that the overall function computed by each neuron remains intact.

More precisely, consider a map $\pi$, parameterzied by $\vh \in \reals^N$, which takes as input a weight configuration $\mcW = (\mWin,\, \mJ,\, \mWout)$, and produces as an output the weight configuration
\begin{align}
    \label{eq:task-preserving-transformation}
    \pi_\vh(\mcW) = (e^{-\mH} \mWin, \: e^{-\mH} \mJ e^{\mH}, \: \mWout e^{\mH}).
\end{align}
Here $\mH$ denotes the diagonal matrix $\mathrm{diag}\{\vh\}$.

It is worth emphasizing some basic properties of the transformation $\pi_\vh$.
The sign of synapses is preserved, and synapses which are initially zero remain so.
Thus, the transformation does not modify the basic wiring diagram of the network---for example by creating or removing synapses, or changing excitatory to inhibitory synapses and vice versa.
Rather, it scales the strength of existing synapses by a positive quantity.
In addition $\pi_\vh$ acts on the recurrent weight matrix $\mJ$ through a similarity transformation.  Thus the entire eigenspectrum of the recurrent weight matrix is conserved under this map. 

However, $\pi_\vh$ satisfies an additional important condition: it exactly preserves the entire input-output output map of the network. 
For this reason, we refer to $\pi_\vh$ as the {task-preserving transformation}.
More formally, for any (finite) value of $\vh$, 
\begin{align}
    \label{eq:task-preserving-transformation-input-output-invariance}
    \mcF(\mcW) = \mcF(\pi_\vh(\mcW)).
\end{align} 
This claim is summarized in Fig. \ref{fig:task-preserving-transformation}, which shows in panels a - f that while the dynamics of the hidden units for two networks related by the task-preserving transformation are different, the output trajectories agree. 
Equation \eqref{eq:task-preserving-transformation-input-output-invariance} is proved in the following proposition.

\begin{proposition}[Task-preserving transformation]
\label{prop:task-preserving-transformation}
The transformation \eqref{eq:task-preserving-transformation} exactly preserves the input-output relationship of the neural dynamics \eqref{eq:neural-dynamics}.
Given two networks receiving the same time course of inputs $\vu(t)$ and with weight configurations $\tmcW$ and $\mcW = \pi_\vh(\tmcW)$ respectively, then:
\begin{enumerate}
    \item If $\tvx(t)$ is the time course of hidden unit neural activity under $\tmcW$, then $e^{-\mH} \tvx(t)$ is the time course of hidden unit neural activity under $\mcW$.
    \item If $\tvy(t)$ is the time course of output neural activity under $\tmcW$, then $\tvy(t)$ is also the time course of output unit neural activity under $\mcW$. 
\end{enumerate}
\end{proposition}

 The proof of proposition 1 is straightforward and can be found in \S\ref{sec:prop1}.

Given an initial weight configuration $\tmcW$, we define the {\it task-preserving manifold} of $\tmcW$ to be the manifold of weight configurations accessible from $\tmcW$ by the task-preserving transformation, i.e., the orbit of $\tmcW$ under the symmetry $\pi_\vh$:
$\{ \pi_\vh(\tmcW): \: \vh \in \reals^N \}$.
The vector $\vh$ provides a set of coordinates on this manifold. 
When referring to the task-preserving manifold we implicitly assume there exists some reference weight configuration $\mcW^0$ at which $\vh=0$.

\section{A measure of robustness to noise in neural activity}
\label{sec:robustness}

We have shown that every weight configuration on the task-preserving manifold computes the same deterministic input-output map.
However, biological systems are inherently noisy \cite{faisal2008noise}, and it has been found that cortical spike trains vary at the sub-millisecond level across trials with identical inputs \cite{tolhurst1983statistical, mainen1995reliability, shadlen1998variable}.
Interestingly, even two networks which compute the same deterministic input-output map may exhibit differing responses to neural noise.
In Fig. \ref{fig:task-preserving-transformation}, we give an example of two networks which perform identically in the deterministic setting but which respond differently when noise is injected to hidden units.

In this section we introduce a quantitative notion of a network's robustness to noise, which we call sensitivity.
This function describes the degree to which noise in neural activity may interfere with the underlying computation of the network.
We connect our notion of sensitivity to the gain of neurons in the network and show that the sensitivity is well-behaved on the task-preserving manifold; in particular, that it possesses a convex geometry in the coordinates $\vh$.

\subsection{Robustness of neural dynamics to random perturbations}
A recurrent network whose dynamics are easily perturbed by small variations in neural activity is unlikely to perform a task robustly in the presence of neural noise.
To capture this notion, we consider the magnitude of the response to a small, isotropic Gaussian perturbation $\vdelta \sim \mathcal N(0, \varepsilon^2 I) $ to the neural dynamics \eqref{eq:neural-dynamics}, averaged over the distribution of states $\mcX$ visited by the network during a task:
\begin{align*}
    \frac{1}{\varepsilon^2} \langle \| \vf(\vx + \vdelta,\vu) - \vf(\vx,\vu) \|_2^2 \rangle_{\vdelta, \vx},
\end{align*}
where for simplicity we have dropped from $\vf$ an explicit dependence on the weights $\mcW$.
We define the {sensitivity} of a network to be the first-order approximation of this quantity,
\begin{align}
    S
    &= 
    \left\langle \left\| \frac{\partial \vf}{ \partial \vx} \right\|_F^2  \right\rangle_{\vx\sim\mcX},
    \label{eq:sensitivity}    
\end{align}
where the approximation is made via Taylor expansion of $\vf$ about $\vx$. 
Networks with lower sensitivity are less easily pushed away from their original trajectories and are therefore more likely to be robust to noise while performing tasks.
In this paper, robustness refers to $S^{-1}$,  the reciprocal of sensitivity.

In the case of the neural dynamics \eqref{eq:neural-dynamics}, the Jacobian matrix takes the form,
\begin{align}
\label{eq:neural-dynamics-jacobian}
\frac{\partial \vf}{ \partial \vx}
&= -\mI + \mJ \,\mathrm{diag}\{\phi'(\vx) \}.
\end{align}
This expression makes use of the gain, $\phi'(x_i)$, of neuron $i$. With respect to the distribution of neural activity $\mcX$, the gain has first and second moments
\begin{align}
    \label{eq:moments-gain-defn}
    \begin{split}
    \mu_i &= \langle \phi'(x_i) \rangle_{x_i\sim\mcX_i} \\
    \sigma^2_i &= \langle \phi'(x_i)^2 \rangle_{x_i\sim\mcX_i}
    \end{split}
\end{align}
for each neuron $i$.
We may rewrite the definition of sensitivity \eqref{eq:sensitivity} using the Jacobian of the neural dynamics \eqref{eq:neural-dynamics-jacobian} in terms of the moments \eqref{eq:moments-gain-defn} to obtain

    \begin{strip} 
\begin{align}
    \label{eq:S-quadratic-in-J}
    \nonumber
    S 
    &=
    \left\langle \left\| \frac{\partial \vf}{ \partial \vx} \right\|_F^2  \right\rangle_{\vx\sim\mcX} \\[3pt]
    \nonumber
    &= \left\langle 
    \Tr \left[ (-\mI + \mJ \,\mathrm{diag}\{\phi'(\vx))(-\mI + \mJ \,\mathrm{diag}\{\phi'(\vx))^T \right]
    \right\rangle_{\vx\sim\mcX} \\[3pt]
    \nonumber
    &= 
    \Tr \left[ 
    \mI - \mJ \,  \left\langle \mathrm{diag}\{\phi'(\vx) \}\right\rangle_\mcX
    - \left\langle  \mathrm{diag}\{\phi'(\vx) \}\right\rangle_\mcX \, \mJ^T
    + \mJ \left\langle \mathrm{diag}\{ \phi'(\vx) \}^2 \right\rangle_\mcX \mJ^T
    \right]
     \\[3pt]
    \nonumber
    &= 
    N - 2\, \Tr \left[ 
     \mJ \,  \left\langle \mathrm{diag}\{\phi'(\vx) \}\right\rangle_\mcX
     \right]
    + \Tr \left[  
    \mJ \left\langle \mathrm{diag}\{ \phi'(\vx) \}^2 \right\rangle_\mcX \mJ^T
    \right]
     \\[3pt]     
    &= 
    \sum_{ij} \sigma^2_j J^2_{ij}  - 2 \sum_i \mu_i J_{ii} + N.
\end{align}
    \end{strip}

The sensitivity, then, can be succinctly written as a function of the recurrent weights and the moments of the neural gains.
However, because the moments $\mu$ and $\sigma^2$ are defined as an average over the distribution of network states that depends---generally intractably---on the weights themselves, sensitivity is not, as it may appear, a straightforward quadratic function of $\mJ$.  
We next show that, perhaps surprisingly, sensitivity turns out to be well behaved as a function on the task-preserving manifold.

\subsection{Sensitivity is a convex function on the task-preserving manifold}

Because networks attaining weight configurations of lower sensitivity might see improved task performance in noisy environments, it is useful to analytically characterize how $S$ varies as a function of network weights.
We have just mentioned the complications of studying this question in general.
Here, we find that $S$ becomes tractable when we consider its variation on the task-preserving manifold, parameterized by the coordinates $\vh$. In particular, we show that $S$ is convex with respect to $\vh$.

\begin{proposition}
\label{prop:sensitivity-convex}
When considered as a function on the coordinates $\vh$ of the task-preserving manifold,
the sensitivity $S$ is a convex function
    \begin{align}
        \label{eq:S-cvx-simple}
        S = \sum_{ij} c_{ij}^0 e^{p(h_j - h_i)} + S_\mathrm{const}
    \end{align}
    for some constants $\{c_{ij}^0\}_{i,j=1}^N$, $p$, and $S_\mathrm{const}$, where $c^0_{ij}\geq0$ for all $i,j$.
\end{proposition}
\begin{proof}

Assume there is some original network $\tmcW = (\tWin, \mJ^0, \tWout)$, whose distribution of neural activity is $\tmcX$, and whose gains have moments $\{\mu_i^0\}_{i=1}^N$ and $\{ (\sigma_i^{0} )^2\}_{i=1}^N$.
Consider a transformed network $\mcW = \pi_\vh(\tmcW)$, with analogous quantities $\mcX$, $\{\mu_i\}_{i=1}^N$, and $\{\sigma^2_i\}_{i=1}^N$. 
From \eqref{eq:task-preserving-transformation}, the recurrent connectivity of the transformed network is $J_{ij} = J_{ij}^0 e^{h_j - h_i}$.
The sensitivity of the transformed network in terms of the task-preserving transformation is obtained by plugging this $J_{ij}$ into \eqref{eq:S-quadratic-in-J}:
 \begin{align}
    \label{eq:S-cvx-intermediate}
    S
    &= \sum_{ij} \sigma^2_j  (\tJ_{ij})^2  e^{2(h_j - h_i)}  - 2 \sum_i \mu_i  \tJ_{ii} + N.
\end{align}

For convexity, it is sufficient to show that the moments of the gain, $\mu_i$ and $\sigma^2_i$, are constant with respect to the task-preserving transformation; i.e., that $\mu_i=\mu^0_i$ and $\sigma^2_i=(\sigma^0_i)^2$ for each neuron $i$.
From Prop. \ref{prop:task-preserving-transformation}, the transformed neural activity can be written in terms of the original rates as $x_i = e^{-h_i} \tx_i$ for all $t>0$.
By assumption, $\phi$ satisfies $\phi(\alpha x) = \alpha \phi(x)$,  and therefore $\phi'(\alpha x) = \phi'(x)$, for all $x \in \reals$ and $\alpha \in \reals^+$.
So $\phi'(x_i) = \phi'( e^{-h_i} \tx_i) = \phi'(\tx_i)$, and
\begin{align*}
\mu_i
= \langle \phi'(x_i) \rangle 
&= \langle \phi'(\tx_i) \rangle 
= \mu^0_i \ \\
\sigma^2_i
= \langle \phi'(x_i)^2 \rangle 
&= \langle \phi'(\tx_i)^2 \rangle 
= (\sigma^0_i)^2
\end{align*}
for each neuron $i$. Therefore the moments of the gains are constant with respect to $\vh$.

We conclude that \eqref{eq:S-cvx-intermediate} takes the form of \eqref{eq:S-cvx-simple}, where the constants are $c_{ij}^0 = (\sigma_j^0)^2 |J_{ij}^0|^2 \geq 0$, $p=2$, and $S_\mathrm{const}= - 2 \sum_i \mu_i^0  \tJ_{ii} + N$.
As \eqref{eq:S-cvx-simple} is a positive linear combination, with constant coefficients, of exponentiated linear functions of $\vh$ (which are convex), then it is in turn a convex function of $\vh$ \cite[Ch. 3]{boyd2004convex}.

\end{proof}

\subsection{Other notions of robustness}
Our definition of sensitivity \eqref{eq:sensitivity} is chosen to emphasize the aspect of recurrent networks which is typically most crucial to task performance; that is, the neural dynamics.
However, there are other notions of robustness which are useful to mention.
Here we briefly address how the task-preserving transformation \eqref{eq:task-preserving-transformation}, parameterized by $\vh$, interacts with several additional relevant notions of sensitivity.

First, consider the sensitivity $S^{\vu\rightarrow \vy}$ of the output trajectory $\{\vy(t)\}$ to fluctuations in the input trajectory $\{\vu(t)\}$.
In our analysis, no choice of $\vh$ can affect this quantity, precisely because the task-preserving transformation exactly preserves the input-output map. Put differently, this notion of sensitivity is intrinsic to the function being computed, not to the manner in which the recurrent network computes it.

Second, consider the sensitivity $S^{\vx\rightarrow \vy}$ of the output trajectory $\{ \vy(t) \}$  to fluctuations in the time course of neural activity $\{ \vx(t) \} $, and respectively, the sensitivity  $S^{\vu\rightarrow \vx}$ of neural activity $\{\vx(t)\}$ to fluctuations in the input trajectory $\{ \vu(t) \}$.
These sensitivities can be individually minimized by taking  $\vh \rightarrow -\infty$ to minimize the norm of $\mWout$ and, respectively, $\vh \rightarrow \infty$ to minimize the norm of $\mWin$.
 The resulting tradeoff between $S^{\vx\rightarrow \vy}$ and $S^{\vu\rightarrow \vx}$ is separate from the problem of minimizing our choice of $S$, and it may be solved independently of the method we now present.

\section{A local learning rule maximizes robustness while preserving task performance}
\label{sec:synaptic-balancing-dynamics}

We have introduced the task-preserving manifold as the set of weight configurations accessible by a symmetry transformation which preserves a given deterministic input-output map.
We have also shown that the sensitivity of neural dynamics to small perturbations in activity---the generally intractable quantity $S$---is well behaved when constrained to the task-preserving manifold.
 A simple question emerges: how might networks traverse the task-preserving manifold to maximize robustness while preserving underlying task performance? 
To address this question, we derive gradient descent dynamics that maximize network robustness in the coordinates of the task-preserving manifold.
We find, perhaps surprisingly, that these dynamics are wholly implementable by biologically plausible local computations within neurons and synapses. 

\begin{figure}
\begin{center}
\includegraphics[width=\linewidth]{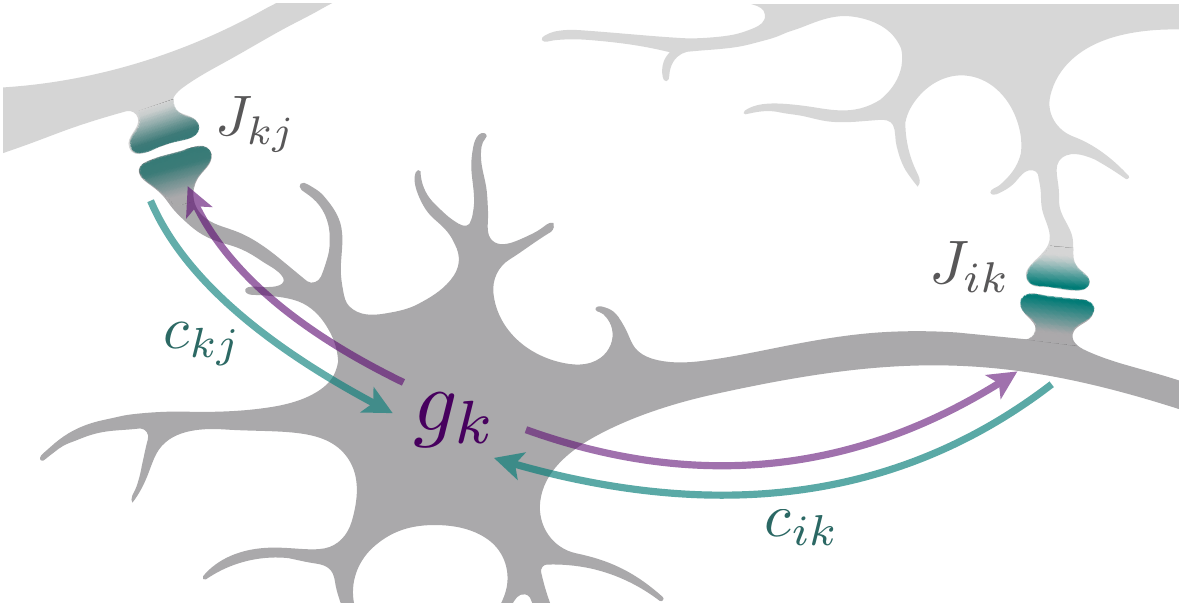}
\caption{
The local computations underlying synaptic balancing.
Each synapse calculates a cost as a function of synaptic strength, as in \eqref{eq:cij-sensitivity}.
Neuron $k$ receives signals of incoming synaptic cost $c_{kj}$ and outgoing synaptic cost $c_{ik}$ and computes the difference $g_{k}$ as in \eqref{eq:gk-costs}. The signal $g_{k}$ then propagates outwards to modify the strength of incoming and outgoing connections, as in \eqref{eq:Jij-dot}, such that the total incoming and total outgoing costs are eventually balanced in every neuron.
}
\label{fig:local-computation}
\end{center}
\end{figure}

From \eqref{eq:S-cvx-simple}, the problem of maximizing robustness $S^{-1}$ is equivalent to that of minimizing the {total cost}:
\begin{align}
    \label{eq:total-cost}
    C
    &=\sum_{ij} c_{ij},
\end{align}
where $c_{ij}$ is the {synaptic cost}
\begin{align}
\label{eq:cij-sensitivity}
c_{ij}
&= \sigma_j^2 |J_{ij}|^p \\
\label{eq:cij-sensitivity-h}
&= c_{ij}^0  e^{p(h_j - h_i)},
\end{align}
with $p=2$, and in the second line we have rewritten the synaptic costs in terms of the initial cost $c^0_{ij} = (\sigma^0_i)^2 |J_{ij}^0|^p$ and the coordinates $\vh$, obtained by writing out \eqref{eq:cij-sensitivity} in terms of the task-preserving transformation $J_{ij} = J_{ij}^0 e^{h_j - h_i}$ \eqref{eq:task-preserving-transformation}.

We call the network connected if the directed graph whose edge weights are given by the initial synaptic cost matrix $\mC^0$ is connected. 
We assume, without loss of generality, that the network is connected (if not, a similar theory applies to each connected component separately). 

We turn to deriving synaptic dynamics that transforms $\vh$ over time so as to maximize robustness on the task-preserving manifold.
Until now we have treated the vector $\vh$ as a free coordinate on the task-preserving manifold.
Henceforth we consider it as a function of time, initializing at the origin and evolving under gradient descent on the total cost.
Similarly, we now view $\mcW = (\mWin, \mJ, \mWout) =  \pi_\vh(\tmcW)$ as a time-varying weight configuration, with initial value $\tmcW$ corresponding to $\vh=0$.

Let the time derivative of $\vh$ be denoted by $\dot \vh = \vg$, which we refer to as the {neural gradient} vector.
The dynamics of gradient descent on $C$ with respect to $\vh$ is given by
\begin{align}
\label{eq:gradient-descent}
\vg = -\gamma \frac{\partial C}{\partial \vh},
\end{align}
where $\gamma \sim \mcO(1)$ is the descent rate.
We assume a separation of timescales, such that neural dynamics \eqref{eq:neural-dynamics} are much faster than the weight dynamics \eqref{eq:gradient-descent}.

The synaptic costs \eqref{eq:cij-sensitivity} obey
\begin{align}
\label{eq:partial-c-partial-h}
\frac{\partial c_{ij}}{\partial h_k}
&= p \, c_{ij} \,(\delta_{kj} - \delta_{ki}),
\end{align}
where $\delta$ is the Kronecker delta. 
In the present regime where $\gamma \sim \mcO(1)$ and  $p \sim \mcO(1)$, we are free to adopt throughout $\gamma p = 1$ for notational convenience. 
Using this, we evaluate the right hand side of \eqref{eq:gradient-descent} via \eqref{eq:total-cost} and \eqref{eq:partial-c-partial-h} to find,
\begin{align}
\label{eq:gk-costs}
g_k
&= \sum_j c_{kj} - \sum_i c_{ik}.
\end{align}
Finally, to see how synaptic strengths are updated, we differentiate the task-preserving transformation $J_{ij} = J^0_{ij}e^{h_j-h_i}$ with respect to time, finding that
\begin{align}
    \label{eq:Jij-dot}
    \dot J_{ij} 
    &= J_{ij}(g_j - g_i), \\
    \label{eq:Win-dot} 
    \dot W^\mathrm{in}_{ik}
    &= - W^\mathrm{in}_{ik} g_i \\
    \label{eq:Wout-dot}
    \dot W^\mathrm{out}_{kj}
    &= W^\mathrm{out}_{kj} g_j 
\end{align}
Equations \eqref{eq:gk-costs} through \eqref{eq:Wout-dot}, along with the definition of synaptic costs \eqref{eq:cij-sensitivity}, are self-contained and collectively comprise the dynamics of our proposed update rule, which we call {synaptic balancing}.

Interestingly, we find that synaptic balancing is entirely implementable by local computations in a network.
First, costs \eqref{eq:cij-sensitivity} are computed at the synapse as a product of the square of the synaptic weight and the presynaptic average gain.
Second, neural gradients \eqref{eq:gk-costs} are computed centrally in the neuron by aggregating and comparing the costs of incoming and outgoing synapses.
Third, synaptic weights \eqref{eq:Jij-dot} are updated in proportion to the difference between the presynaptic and postsynaptic neural gradients.

In this scheme we assume that the synaptic weights, synaptic costs, and neural gradients can be represented as biophysical quantities inside neurons and synapses.  We suppose, as depicted in Fig. \ref{fig:local-computation}, that in one direction a neuron traffics its gradient from the soma to incoming and outgoing synapses, and, in the other, it traffics synaptic costs from synapses to the soma.
Finally, we assume that each synapse is able to measure and store certain statistics of presynaptic activity, namely, the average presynaptic gain.
This assumption is in line with other models of synaptic plasticity which introduce some form of leaky integration of neural activity, e.g. the sliding threshold of the BCM rule \cite[p. 288]{dayan2001theoretical}.

Notably, although the coordinates $\vh$ and the task-preserving transformation $\pi_\vh$ are of instrumental value in deriving and analyzing our rule, the biological procedure just described implements synaptic balancing without an explicit representation of $\vh$.

Next, we study the equilibrium state of synaptic balancing and offer a simple condition under which an equilibrium exists, thereby guaranteeing the stability of our rule.

\begin{figure*}
\makebox[\textwidth][c]{
\includegraphics[width=6in]{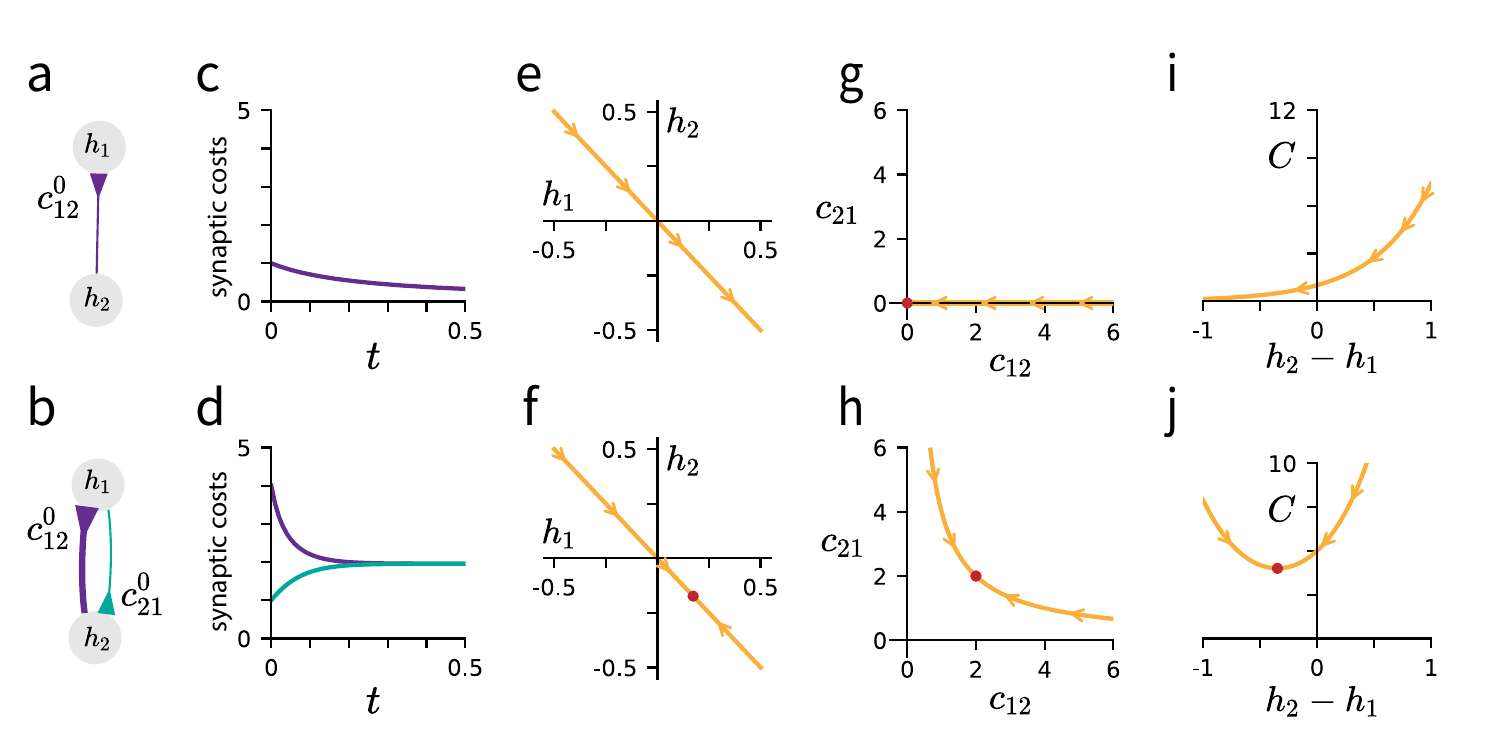}
}
\caption{
Network topology determines the geometry of the task-preserving manifold and the dynamics of synaptic balancing.
Top row: network with two hidden units connected by a single synapse.
Bottom row: network with two hidden units connected reciprocally. 
\textbf{(a-b)} 
Network diagrams showing topology and initial synaptic costs, indicated by line thickness.
Input and output neurons are not shown.
\textbf{(c-d)}
Trajectory of synaptic costs over the course of synaptic balancing.
Line colors match synapse colors in panels (a-b).
Panel (c) matches \eqref{eq:two-neuron-solution-feedforward} and panel (d)  matches \eqref{eq:two-neuron-solution-connected}.
\textbf{(e-f)}
The feasible set of $\vh$ satisfying \eqref{eq:sum-hi-zero},
with flow lines indicating trajectory of synaptic balancing.
Panel (f): Red point indicates the (finite) minimizer of the total cost.
\textbf{(g-h)}
Tradeoff between $c_{12}$ and $c_{21}$ as a function of $h_2-h_1$, with red point indicating the optimal value of total cost $C = c_{12}+c_{21}$.
\textbf{(i-j)}
Total cost $C$ as a function of $h_2-h_1$.
(i) Optimal cost is zero, attained at an infinite value of $\vh$.
(j) Optimal cost is positive, attained at finite $\vh$.
}
\label{fig:task-manifold-geometry}
\end{figure*}

\section{Strongly connected networks attain balanced equilibrium}
\label{sec:equilibrium-existence-stability}

Synaptic balancing asymptotically converges to a global minimum on the task-preserving manifold, by construction, because it is gradient descent on the convex function $C$.
However, a global minimum may not exist at any finite value of the task-preserving manifold coordinates $\vh$.
For example, the convex function $f(h) = e^h$ on the real line possesses no global minimum at finite $h$, and the asymptotic convergence of gradient descent is not assured.
It is therefore important to establish the criteria under which we expect synaptic balancing to stably converge to a (finite) minimizer $\vh^*$.
To do this, we provide two simple criteria:
first, a balance condition which describes the set of equilibria of synaptic balancing, and second, a strongly-connected condition which implies the existence of an equilibrium on the task-preserving manifold that attains exponential convergence.

A weight configuration globally minimizes the total cost on the task-preserving manifold if and only if it satisfies the balance condition
\label{prop:balance-condition}
\begin{align}
    \label{eq:balance-condition}
    \sum_i c_{ik} = \sum_j c_{kj}
\end{align}
for each neuron $k$.
To see this, observe that because the total cost $C$ is convex and differentiable as a function of $\vh$, a coordinate $\vh^*$ achieves a global minimum if and only if $\frac{\partial C}{\partial \vh}(\vh^*)=0$, which, from \eqref{eq:gradient-descent}, holds if and only if $\vg = 0$.
Solving for $\vg=0$ in \eqref{eq:gk-costs} yields the {balance condition} \eqref{eq:balance-condition}.

It is because of \eqref{eq:balance-condition}
that we use the term {balancing} to describe our rule:
the total cost is minimized, and a weight configuration is a stable fixed point of the weight dynamics \eqref{eq:Jij-dot}, if and only if the total synaptic costs of each neuron's incoming and outgoing synapses are equal.
It remains unclear, however, whether the balance condition \eqref{eq:balance-condition} is attainable in every case, or put differently, whether synaptic balancing is stable from every initial condition.

We now state a simple topological criterion which (we will show) implies the stability of synaptic balancing.
A recurrent network is said to be strongly connected if between every pair of neurons $i$ and $j$ there exists a path of synapses, each with positive synaptic cost, from $i$ to $j$.
Conceptually, this condition simply means that the recurrent connectivity does not possess any embedded directed acyclic structure, such as purely feedforward connections.

With regard to biology, it is of course generally challenging to show that any given biological network is strongly connected.
With that said, we believe that the highly recurrent nature of the central nervous system suggests that this assumption may not be far from reality in certain cortical regions.
Long-range feedback projections as well as local recurrence in cortex suggest that non-negligible sub-networks of neurons participating in, for example, the ventral visual stream are indeed strongly connected \cite{kravitz2013ventral, nayebi2021goal}.
Such networks are strong candidates for synaptic balancing.

Mathematically, it is known that the connected components of any balanced matrix, i.e., a matrix satisfying \eqref{eq:balance-condition}, are strongly connected \cite{hooi1970class}, and that, conversely, such matrices may be made balanced through a positive diagonal similarity transformation \cite{eaves1985line}.  

 Before formally describing the stability of our rule, we make the preliminary observation that the synaptic costs in \eqref{eq:cij-sensitivity-h} are invariant under the transformation $\vh \mapsto \vh + b$, for $b \in \reals$.
Therefore the gradient of the total cost $C$ is perpendicular to the all ones vector $\vone$, and gradient descent on $C$ does not explore that direction:
From \eqref{eq:gk-costs}, $\vone^T \dot \vh = \sum_i g_i = 0$.
Assuming as before that $\vh^0=0$, synaptic balancing exclusively yields solutions that satisfy
\begin{align}
\label{eq:sum-hi-zero}
    \vone^T \vh = \vone^T \vh^0 = 0.
\end{align} 

We may now connect these remarks to our weight dynamics, showing that synaptic balancing converges exponentially to a balanced configuration whenever the initial cost matrix is strongly connected.
 Fig. \ref{fig:task-manifold-geometry} illustrates this point in two-neuron networks, showing the time course of synaptic balancing and the geometry of the cost function. In the first case of a single feedforward connection, a stable equilibrium is not reached at any finite $\vh$. In the second case of recurrent connections but asymmetric initial synaptic costs, the network converges to a stable equilibrium at finite $\vh$ and positive total cost $C$, with symmetric final costs.   
A short proof is provided in the following proposition.

\begin{proposition}
\label{prop:topology-convergence}
In a strongly connected network, synaptic balancing converges to a globally exponentially stable equilibrium, which is the unique solution to \eqref{eq:sum-hi-zero} and \eqref{eq:balance-condition} on the task-preserving manifold.
\end{proposition}
\begin{proof}[Proof (Sketch)] 
Because synaptic balancing does not increase the total cost $C = \sum_{ij} c_{ij}$ \eqref{eq:total-cost}, and because the synaptic costs are nonnegative, we have that $c_{ij} \leq C^0$ along the trajectory of synaptic balancing for all $i,j$.
If $c^0_{ij} >0$, then since $c_{ij} = c_{ij}^0 e^{p(h_j - h_i)}$ \eqref{eq:cij-sensitivity-h}, it follows directly that $e^{h_j - h_i} \leq (C^0/c_{ij}^0)^{1/p} < \infty$.
If $c^0_{ij} = 0$, then an ``indirect'' upper bound for $e^{h_j - h_i}$ is obtained by collapsing the direct upper bounds along a path of positive synaptic costs from $j$ to $i$:
In the case of a single intervening neuron $k$, for example, $e^{h_j - h_i} = e^{h_j - h_k} e^{h_k -h_i} \leq (C^0/c_{kj}^0)^{1/p} (C^0/c_{ik}^0)^{1/p}$. 
By the strong connectedness assumption, such a path exists between every $i$ and $j$, so we may in this way obtain upper bounds on $ h_j - h_i$, and therefore on $|h_j - h_i|$, for all $i,j$.
Combined with the constraint $\sum_i h_i = 0$ \eqref{eq:sum-hi-zero}, it follows that $\vh$ is contained in a compact set over the trajectory of synaptic balancing. 
We conclude that the minimizer of $C$, a convex function of $\vh$, is attained in this set.
To show that the minimizer is unique and globally exponentially stable, it is sufficient to show that under strong connectedness, $C$ is strongly convex on the sublevel set of $C^0$.  This is shown in   \S \ref{sec:equilibrium-existence-condition-appendix} with a more detailed proof. 
\end{proof}

In summary, synaptic balancing consists of local synaptic updates which are stable in any recurrent network that does not possess directed acyclic structure.
By minimizing a convex cost on the task-preserving manifold, the rule maximizes the robustness of neural dynamics to noise while maintaining underlying task performance.
In the following section we will mention a few useful generalizations of our model.

\section{Generalizations and a connection to Lax dynamical systems}
\label{sec:generalizations-lax}

We have derived synaptic balancing as gradient descent on a behaviorally-relevant convex function---the sensitivity of the network to neural noise---and characterized the stability and equilibria of the rule.
We now observe that our framework for synaptic balancing admits several generalizations: first, a broader class of synaptic costs, of which sensitivity is a particularly salient example, and second, a yet more general class of matrix-valued dynamical systems known as Lax dynamics, which have found widespread applications in fields from physics \cite{lax1968integrals} to optimization and numerical linear algebra, where they are known as isospectral flows \cite{helmke1994optimization, chu2008linear}. 

\subsection{A general class of well-behaved synaptic costs}

The expression for synaptic costs \eqref{eq:cij-sensitivity} was derived to maximize robustness by minimizing the behaviorally relevant sensitivity function \eqref{eq:sensitivity}.
Generally, however, the framework presented here admits a broad class of synaptic cost functions.
If the total cost is the sum of synaptic costs, as in \eqref{eq:total-cost},
then the synaptic cost $c_{ij}$ may be an arbitrary fixed function of $J_{ij}$.
Importantly, the stability of the resulting weight dynamics, and the optimality of any equilibria, are not guaranteed in general.

A sensible class of synaptic costs which are convex in $\vh$ are the {power-law costs}
\begin{align}
    \label{eq:cij-power-law}
    c_{ij} = \alpha_{ij} |J_{ij}|^p,
\end{align}
where $\alpha_{ij} \geq 0$ for all $i,j$, and $p>0$. This definition includes sensitivity \eqref{eq:cij-sensitivity} as a special case, with $\alpha_{ij} = \sigma^2_{j}$ and $p=2$.
Other costs of the power-law form include the $\ell_1$ and $\ell_2$ matrix penalties studied in statistics and machine learning
\cite[Ch.\ 3.4]{hastie2009elements}.
Notably, the power-law costs \eqref{eq:cij-power-law}:
\emph{i)} obey \eqref{eq:cij-sensitivity-h}, \emph{ii)} correspond to neural gradients of the form \eqref{eq:gk-costs}, and \emph{iii)} inherit the stability result of Prop. \ref{prop:topology-convergence},
 extending the findings of this work to a broader class of cost function.
Robustness, specifically, is only maximized with the particular choice of costs \eqref{eq:cij-sensitivity}.
 
\subsection{Synaptic Lax dynamics}
A yet further generalization of our synaptic update framework dispenses with a cost function entirely.
A {Lax dynamical system} is an evolution on an $N \times N$ matrix $\mA$ such that 
\begin{align}
\label{eq:lax-dynamical-system-defn}
\dot \mA
=[\mA, \, k(\mA)];
\quad
\mA(0)
=\mA_0,
\end{align}
where $k: \reals^{N \times N} \longrightarrow \reals^{N \times N}$ is a matrix-valued function which depends on time only through its argument $\mA$, and $[\cdot,\, \cdot]$ denotes the Lie bracket, which is defined as
\begin{align*}
    [\mA, \, \mB]  &=\mA \mB - \mB \mA.
\end{align*}
An important property of Lax dynamics is that the evolution of $\mA$ takes the form of a time-varying similarity transformation.
To see this, consider the $N \times N$ matrix $\mK$ which evolves in time according to 
\begin{align*}
    \dot \mK = \mK k(\mA); \quad \mK(0) = \mI.
\end{align*}
It is straightforward to show that $\mK(t)$ is nonsingular for all $t$ and that the similarity transformation
\begin{align*}
    \mA(t) = \mK(t)^{-1} \mA_0 \mK(t)
\end{align*}
is the unique solution of \eqref{eq:lax-dynamical-system-defn}.
As a consequence, Lax dynamical systems conserve the entire spectrum of the matrix $\mA$ and are sometimes referred to as {isospectral flows}.

Synaptic balancing is a specific instance of Lax dynamics.
We may rewrite \eqref{eq:Jij-dot} as 
\begin{align}
    \label{eq:J-dot}
    \dot \mJ = [\mJ, \, \diag\{\vg\}], %
\end{align}
noting that the elements of $\vg$ depend on time only through a dependence on the elements of $\mJ$. Thus this dynamics is a special case of the general definition \eqref{eq:lax-dynamical-system-defn}.

The Lax dynamics of synaptic balancing \eqref{eq:J-dot} encompass a very general form of task-preserving local learning rule.
If the neural gradient $g_k$ is allowed to depend arbitrarily on the incoming and outgoing weights of neuron $k$, then synaptic Lax dynamics remains confined to the task-preserving manifold and is fully locally computable in the sense of Fig. \ref{fig:local-computation}.
Further, all quantities which are conserved on the task-preserving manifold---including the spectrum of the recurrent weight matrix and the product of synaptic weights along every directed closed loop---are conserved by the synaptic Lax dynamics as well.

Synaptic Lax dynamics, if chosen in such a way to be stable, might be used to regulate any number of structural properties of the network.
Any equilibrium of the synaptic Lax dynamics satisfies the equality condition 
$g_i = g_j$
for all pairs $(i,j)$, which generalizes the balance condition \eqref{eq:balance-condition}.
For choices of neural gradient for which this equilibrium is attainable and stable, the synaptic Lax dynamics offer a flexible framework for adapting network weights without affecting task performance.
For example, \cite{schneider1991max} studies the problem of stably balancing the {maximum} incoming and outgoing synaptic strengths, and
\cite{rothblum1992scalings} studies the problem of balancing the incoming and outgoing {products} of synapses.
Each of these aims could be implemented by the synaptic Lax dynamics through proper choice of $g_k$.
Our work unifies these efforts under a general framework of continuous-time weight dynamics \eqref{eq:J-dot}, and, unlike previous approaches,
derives from functional considerations a particular form of neural gradient
that provably minimizes a behaviorally-relevant cost function on the task-preserving manifold.

\section{Regularized networks are nearly balanced}
\label{sec:regularized-networks-balanced}

\begin{figure*}
\makebox[\textwidth][c]{
\includegraphics[width=6in]{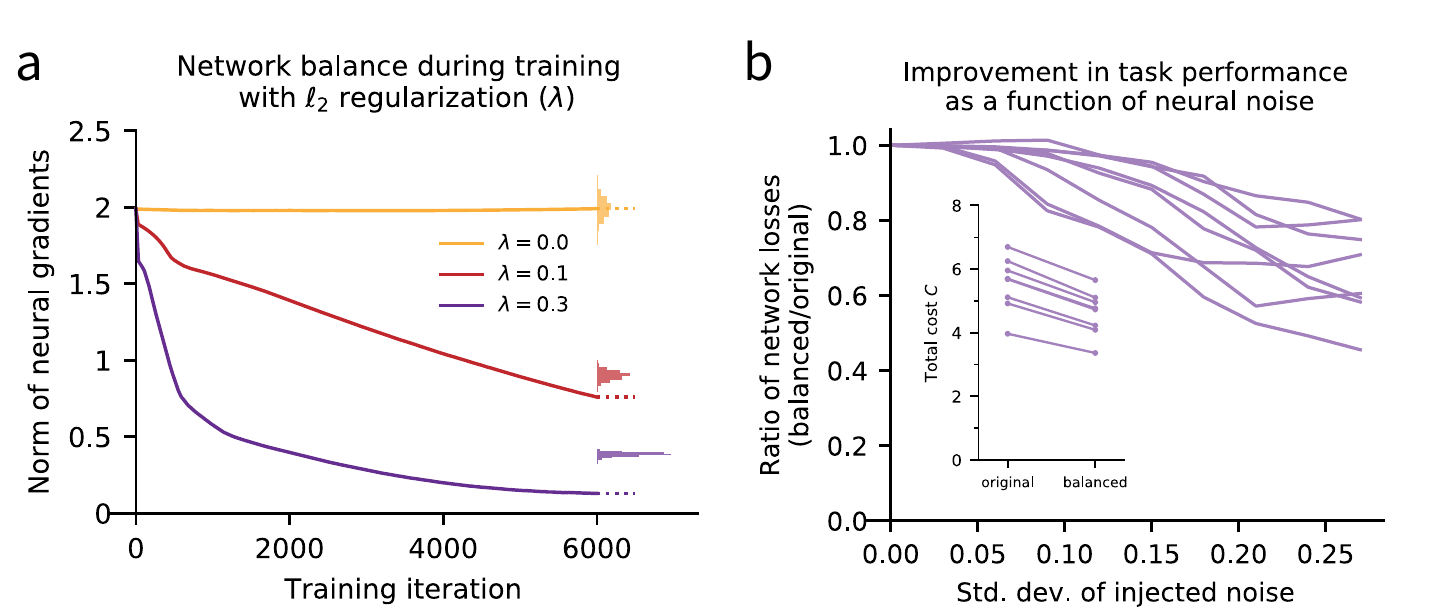}
}
\caption{
The balance condition in trained networks.    
\textbf{(a)} 
Networks with $N=256$ neurons are trained via gradient descent on a context-dependent integration (CDI) task modeled after \cite{mante2013context}, with varying levels of $\ell_2$ penalty $\lambda \sum_{ij} J_{ij}^2$.
The norm of neural gradients, $\|\vg \|$, is shown over the course of training.
When $\lambda = 0$, neural gradients are fixed by gradient descent dynamics.
When $\lambda > 0$, trained solutions tend towards the balance condition \eqref{eq:balance-condition}, with $c_{ij}=J_{ij}^2$.
Histograms at right denote the empirical null distribution of $\|\vg \|$ under permutation of the rows of $\mC$ at the final training iteration. For positive values of $\lambda$, the actual value of $\|\vg \|$ (dotted line) falls significantly below the null distribution. 
\textbf{(b)}
Synaptic balancing with the robustness cost function \eqref{eq:cij-sensitivity} is applied to several networks trained with $\lambda=0.3$ (corresponding to the purple curve in panel \textbf{a}).
The original and balanced networks are run on the CDI task at varying levels $\varepsilon$ of Gaussian noise injected into hidden dynamics.
For each network pair the ratio is plotted of task loss of the balanced network to that of the original network, as a function of $\varepsilon$.
As dynamics become more noisy, the performance of the original networks (as measured by loss on the task) degrades faster than that of the corresponding balanced networks. Inset: total cost $C$ of the original vs. balanced networks \eqref{eq:total-cost}.\text
}
\label{fig:trained-networks}
\end{figure*}

Our focus so far has been on establishing theoretical links between the balance condition, noise-robust computation and the proposed local plasticity rule.
We now turn to studying the interaction of these concepts with other forms of plasticity, such as task-relevant learning.
In this section we present evidence that commonly studied classes of network are, in fact, generically in the vicinity of equilibrium, and that the balance condition \eqref{eq:balance-condition}, far from being an obscure edge case of network configurations, is in some sense a generic state.
In particular, we find that the balance condition is approximately attained by any learning algorithm minimizing a very general class of regularized loss functions.

Recall that we denote by $\mcF(\mcW)$ the input-output map of the network \eqref{eq:input-output-map}, which takes input trajectories $\{\vu(t)\}$ to output trajectories $\{\vy(t)\}$ and which is parameterized by the weight configuration $\mcW = (\mWin, \mJ, \mWout)$. 
Consider an arbitrary loss function $L$ on the input-output map $\mcF$.
We are not concerned with the particular functional form of $L$---just that it depends on the weight configuration only through the input-output map.
For example, the loss might measure the performance of the network on some task, with weight configurations that attain lower loss achieving better task performance.

We consider a regularized loss function $\Lreg$ which is the sum of the loss $L$ and element-wise regularization of the recurrent weight matrix:
\begin{align}
    \label{eq:regularized-loss}
    \Lreg(\mcW) = L\left(\mcF\left(\mcW\right)\right) +  \sum_{ij} \alpha_{ij} |J_{ij}|^p, %
\end{align}
where $p, \alpha_{ij} > 0$ for all $i,j$.
The expression for $\Lreg$ encompasses both $\ell_1$ and $\ell_2$ regularization of the recurrent weight matrix, for example, by setting $p=1$ and $p=2$ respectively.

\begin{proposition}
\label{prop:regularized-universality}
Suppose $\Lreg$ has a local minimum at the weight configuration $\mcW^*$. Then $\mcW^*$ satisfies the balance condition \eqref{eq:balance-condition}, with synaptic costs $c_{ij} = \alpha_{ij} |J_{ij}|^p$.
\end{proposition}
\begin{proof}
We provide a proof by contradiction. Assume that there exists a weight configuration $\tmcW$ which is a local minimum of $\Lreg$ but which is not an equilibrium of synaptic balancing.
Let $\vh(t)$, $t\geq0$, describe the trajectory of synaptic balancing initialized at $\tmcW$, with synaptic costs $c_{ij} = \alpha_{ij} |J_{ij}|^p$.
Because synaptic balancing preserves the input-output map $\mcF$,
the loss function $L(\mcF(\pi_{\vh(t)}(\tmcW)))$ is constant as a function of $t$, i.e.,
\begin{align}
    \label{eq:loss-constant-under-synaptic-balancing}
   \frac{d}{dt}  L(\mcF(\pi_{\vh(t)}(\tmcW))) = 0.
\end{align}
Because we have assumed that $\mcW^0$ is not an equilibrium of synaptic balancing, the total cost $C = \sum_{ij}  \alpha_{ij} |J_{ij}|^p$ is strictly decreasing at $t=0$, i.e.,
\begin{align}
    \label{eq:cost-decreasing-under-synaptic-balancing}
    \frac{d}{dt} C(\pi_{\vh(t)}(\tmcW)) < 0.
\end{align}
Combining \eqref{eq:loss-constant-under-synaptic-balancing} and \eqref{eq:cost-decreasing-under-synaptic-balancing}, we have that
\begin{align*}
\begin{split}
    \frac{d}{dt} \Lreg(\pi_{\vh(t)}(\tmcW)) 
    =&\frac{d}{dt} L(\mcF(\pi_{\vh(t)}(\tmcW))) \\ & +\frac{d}{dt}C(\pi_{\vh(t)}(\tmcW)) \\
    <& 0
    \end{split}
\end{align*}
for all $t$ and in particular, $t=0$.
Thus, $\Lreg$ is a decreasing function on the trajectory of synaptic balancing initialized at $\tmcW$, and $\tmcW$ is not a local minimum of $\Lreg$.
We conclude that every local minimum of $\Lreg$ satisfies the balance condition \eqref{eq:balance-condition}.
\end{proof}
An inverse result---that in unregularized networks trained with gradient descent, the initial (generally nonzero) neural gradients are exactly conserved by training---has been noted in, e.g., \cite{tanaka2020pruning, du2018algorithmic}.

In practice, numerical optimization algorithms for network training terminate before achieving true local minima, so we do not expect every network trained with regularization to exactly exhibit the balance condition \eqref{eq:balance-condition}.
In Fig. \ref{fig:trained-networks}a we show, however, that trained, $\ell_2$-regularized networks exhibit neural gradients that are substantially lower than would be attained by chance, under random perturbation of rows of the weight matrix.

In understanding these results, the reader should keep in mind that under our terminology, a network may be balanced with respect to one cost function (say, the $\ell_2$ cost), but not balanced with respect to another cost function (say, the robustness cost).
Similarly, the form of regularization that is used during training affects the balance properties of the local minimum reached; a network trained with $\ell_2$ regularization need not satisfy the robustness balance condition.

Nonetheless, we have showed that any regularized network will become exactly balanced (in some sense) through training.
This observation suggests that attaining the equilibrium condition of synaptic balancing is not only functionally desirable from a robustness standpoint but also is a generic consequence of learning rules which optimize a loss function of the form \eqref{eq:regularized-loss}.

\section{Balancing empirically improves task performance}
\label{sec:trained-network-improvement}

To experimentally measure the effect of synaptic balancing on task performance, we trained networks via gradient descent on a context-dependent integration task modeled after \cite{mante2013context}.
This process yielded a trained network which adequately performed the task.
We applied synaptic balancing to the trained network, using the robustness cost function \eqref{eq:cij-sensitivity}, and stored the equilibrium weight configuration as our balanced network.
We simulated the trajectories of the original (trained) versus balanced network and observed task performance while varying the levels of additive Gaussian noise in the neural dynamics \eqref{eq:neural-dynamics}.
Details of the task and network training are provided in \S\ref{sec:network-training-details}.

As expected, both the original and balanced networks perform identically in the absence of noise, due to the task-preserving nature of synaptic balancing.
We also found that the task performance, as measured by the trial-averaged task loss on held-out test data, deteriorated in both the trained and balanced networks as the noise level increased.
However, this deterioration was noticeably attenuated in the balanced networks, and
trained networks that had not been balanced proved more sensitive to higher levels of Gaussian noise.
The decay in relative performance of original versus balanced networks is illustrated across several network instantiations in Fig. \ref{fig:trained-networks}b.

The improvement exhibited by the balanced network is remarkable in part because synaptic balancing does not make use of a task-specific error signal, but merely uses summary statistics of the average neuronal gain during the task.
These results confirm that in networks which are already performing a task, the variability of neural responses can be suppressed simply by shifting synaptic weights towards a balanced configuration via the task-preserving transformation.  As noted above, this task preserving transformation can be implemented by local synaptic learning rules that require no knowledge of the task. 

\section{Exact and approximate trajectories of synaptic balancing}
\label{sec:exact-approx-solutions}

We have shown that the ordinary differential equation \eqref{eq:Jij-dot} is a member of a widely studied class of dynamical systems called Lax dynamics.
In the interest of better understanding the action of our dynamics on the recurrent weight matrix as a whole, we now turn to closed-form solutions to \eqref{eq:Jij-dot} that specify the evolution of synaptic balancing over time.
When the network has just two neurons, our solution is exact;
in general, we derive a quadratic approximation taking the form of a heat equation.

\subsection{Exact trajectory with two neurons}

In a two-neuron network, the balancing dynamics admit an exact analytical solution for the time course of the synaptic costs $c_{12}$ and $c_{21}$, when these costs take the power-law form \eqref{eq:cij-power-law}. 
If both initial synaptic costs are positive, we find that they evolve as
\begin{align}
\label{eq:two-neuron-solution-connected}
\begin{split}
c_{12}(t)
&=  \hat{c}_{12} \, q(t)\\
c_{21}(t)
&=  \hat{c}_{12} \, q(t)^{-1},
\end{split}
\end{align}
where
\begin{align}
\label{eq:cij-geometric-mean}
\hat c_{ij}&=\sqrt{c_{ij}^0 c_{ji}^0}
\end{align}
and
\begin{align*}
    q(t) =  \tanh{\left[
    2p^{2} (\hat{c}_{12}) t 
    + \tanh^{-1} ({c_{12}^0 / c_{21}^0})^{1/2} \right]}
\end{align*}
If just a single initial synaptic cost is positive---suppose it is $c_{12}$--- then it evolves as
\begin{equation}
\label{eq:two-neuron-solution-feedforward}
c_{12}(t) = \frac{c_{12}^{0}}{2 c_{12}^{0} \gamma p^{2}\, t  + 1}
\end{equation}
while $c_{21}$  is fixed at zero.
In agreement with Prop. \ref{prop:topology-convergence}, the latter case suffers from unbounded growth of the input and output weights over time as $J_{12} \rightarrow 0$, and there is no stable equilibrium of the dynamics.
The trajectories in \eqref{eq:two-neuron-solution-connected} and \eqref{eq:two-neuron-solution-feedforward} are illustrated in the final panel of Fig. \ref{fig:task-manifold-geometry}.

\subsection{Approximate trajectory via the heat equation}
Except in the two-neuron scenario, we are not aware of a general analytic solution to the trajectory of synaptic balancing.
Here we show that a closed-form approximate solution is attainable in general, however, shedding light on the macroscopic patterns of weight modification that synaptic balancing induces in a network.
Our approach is to consider the time evolution of the neural gradients, which we tie to solutions of the heat equation on a graph.
This then yields a closed-form approximate local solution for the trajectory of the coordinates $\vh$.

The time evolution of the neural gradients \eqref{eq:gk-costs} is given by
\begin{align}
    \label{eq:g-dot-jacobian}
    \dot \vg
    &= \frac{\partial \vg}{\partial \vh} \vg.
\end{align}
Since $\vh$ evolves under gradient descent \eqref{eq:gradient-descent}, the Jacobian of the dynamics of $\vh$ is a sign-flipped version of the Hessian of the cost function \eqref{eq:total-cost},
\begin{align}
    \label{eq:h-jacobian-hessian}
    \frac{\partial \vg}{\partial \vh}
    &= -\gamma \frac{\partial^2 C}{\partial \vh^2}.
\end{align}
A short calculation via \eqref{eq:partial-c-partial-h} finds the Hessian to be
\begin{align}
    \label{eq:C-hessian-laplacian}
    \frac{\partial^2 C}{\partial \vh^2}
    &= p^2 \mL, 
\end{align}
where $\mL$ takes the form of a Laplacian matrix corresponding to a graph with edge weights given by what we call the conductance matrix $\bmC$, which has $ij$th element
\begin{align}
    \label{eq:cij-bar-defn}
    \bar c_{ij} = c_{ij} + c_{ji},
\end{align}
and which is evidently symmetric.
Explicitly, the Laplacian $\mL$ has elements drawn from $\bmC$ as follows:
\begin{align}
    \label{eq:L-graph-Laplacian}
    L_{ij}
    &=
    \left\{
    \begin{array}{ll}
         -\bar{c}_{ij}, &i\neq j, \\
         \sum_{k\neq i} \bar c_{ki}, &i=j.
    \end{array}
    \right.
\end{align}

Combining \eqref{eq:g-dot-jacobian}, \eqref{eq:h-jacobian-hessian}, and \eqref{eq:C-hessian-laplacian}, and maintaining $\gamma=p^{-1}$ as before, the time derivative of the neural gradients is
\begin{align}
    \label{eq:g-dot-heat-equation}
    \dot \vg
    &= -p \mL \vg.
\end{align}
Equation \eqref{eq:g-dot-heat-equation} is the {heat equation} of the graph corresponding to $\mL$, in analogy to the heat equation in continuous environments \cite{chung1997spectral}. 
This suggests an intuitive physical description, and approximate mathematical solution, of the evolution of neural gradients.

Since $\mL$ is itself time-varying, the dynamics \eqref{eq:g-dot-heat-equation} do not admit a straightforward exact solution.
However, in the vicinity of a weight configuration $\tilde \mcW$, we may take $\mL$ to be fixed at $\tilde\mL$ to obtain a closed-form expression which approximates the exact solution to \eqref{eq:g-dot-heat-equation}.
This corresponds to gradient descent dynamics on the quadratic Taylor approximation of $C$ in $\vh$ evaluated at $\tilde\mcW$.
Under fixed $\tilde \mL$, \eqref{eq:g-dot-heat-equation} is a (time-invariant) linear dynamical system, and
the solution at time $t$ may be expressed in terms of the {heat kernel} $e^{-p \tilde \mL t}$:
\begin{align}
\label{eq:g-solution-equilibrium}
\vg(t)
&= e^{-p \tilde \mL t} \vg^0 \nonumber \\
&= \sum_{i: \lambda_i > 0} e^{-p \lambda_i t} \vv^i (\vv^i)^T \vg^0,
\end{align}
where $\vg^0$ denotes the neural gradient at time $t=0$, and $\lambda_i$, $\vv^i$ denote the $i$th eigenvalue and eigenvector of $\tilde \mL$.
The sum is taken over all $i$ such that $\lambda_i$ is positive, since if $\lambda_i = 0$ for some $i$ then $\vv^i$ is an indicator vector of a connected component, which must be orthogonal to any realizable neural gradient $\vg^0$.

Integrating \eqref{eq:g-solution-equilibrium} yields a closed-form approximate expression for the evolution of the coordinates $\vh$ on the task-preserving manifold:
\begin{align}
\label{eq:h-solution-equilibrium}
\vh(t)
&=  \sum_{i:\lambda_i>0} (p\lambda_i)^{-1} (1-  e^{-p \lambda_i t}) \vv^i (\vv^i)^T \vg^0,
\end{align}
where we have chosen boundary conditions such that $\vh(0) =0$.
As $t \rightarrow \infty$, the network reaches an equilibrium  $\vh \rightarrow \vh^*$, given by
\begin{align}
\label{eq:hstar-solution-equilibrium}
\vh^*
&= p^{-1} \tilde \mL^\dagger \vg^0.
\end{align}
The notation $\tilde \mL^\dagger$ denotes the Moore-Penrose pseudoinverse of $\tilde\mL$.

The approximate closed-form dynamics derived in this section, and especially the role of the Laplacian matrix $\mL$, help shape intuition about the behavior of our rule:
Neural gradients diffuse in a network akin to heat diffusing on a graph, with higher spatial-frequency modes of the graph decaying more quickly than lower spatial-frequency modes.
As illustrated in Fig.  \ref{fig:heterosynapic-plasticity-perturbation}, numerical simulations verify that the quadratic approximate dynamics presented here accurately describe synaptic trajectories near equilibrium.  
{For a 12-neuron, ring topology network, Fig. \ref{fig:heterosynapic-plasticity-perturbation} shows how a perturbation to a synaptic cost is redistributed throughout the network by synaptic balancing dynamics, as well as the evolution of $\vh$ in the Laplacian matrix eigenmode basis.}

\subsection{General bounds on equilibrium cost}
We have provided exact and approximate closed-form dynamics of synaptic balancing.
To complement these results, we now turn to results on the optimal value of the cost function, stated in terms of the initial cost matrix $\mC^0$.

\begin{proposition}
\label{prop:bounds}
Suppose synaptic costs are of the power-law form \eqref{eq:cij-power-law}.
Then given initial neural gradients $\vg^0$ and initial total cost $C^0$, the minimum total cost $C^*$ on the task-preserving manifold is bounded above by
\begin{align}
    \label{eq:C-star-upper-bound}
     C^*
    &\leq C^0 - \frac{1}{8 C^0} \|\vg^0 \|_2^2,
\end{align}
and below by
\begin{align}
    \label{eq:C-star-lower-bound-c-hat}
    \sum_{ij} \hat c_{ij}
    \leq C^*,
\end{align}
where $\hat{c}_{ij} = \sqrt{c_{ij}^0 c_{ji}^0}$ for all $i,j$,  as in \eqref{eq:cij-geometric-mean}.

\end{proposition}

A proof is given in  \S\ref{sec:analytic-bounds}.
Intuitively, the upper bound \eqref{eq:C-star-upper-bound} says that the greater the initial deviation from the balance condition (in the form of large neural gradients), the greater the guarantee that synaptic balancing will improve the total cost. 
The lower bound \eqref{eq:C-star-lower-bound-c-hat} says that the more symmetric the initial cost matrix, the less that synaptic balancing is able to improve the costs.

Recall that in a two-neuron network, the equilibrium cost matrix is symmetric \eqref{eq:two-neuron-solution-connected}.
In fact, we find that that symmetry is preferred by synaptic balancing, and
the equilibrium cost matrix will be symmetric in any circumstance in which symmetry is attainable on the task-preserving manifold.
If $\mC^0$ is the initial matrix of synaptic costs, and $\mC^0$ is positive diagonally symmetrizable, then the minimum is attained at the cost matrix $\hat \mC =e^{-\mH^*} \mC^0 e^{\mH^*}$, whose $i,j$th element is the geometric mean of initial costs \eqref{eq:cij-geometric-mean}, and equality is attained in \eqref{eq:C-star-lower-bound-c-hat}.

As an example of symmetric equilibrium beyond the case of $N=2$, consider a rank-one network whose initial costs factor as $\mC^0 = \va \vb^T$.
This network is strongly connected only if every element of $\va$ and $\vb$ is positive, in which case the equilibrium cost matrix is $\mC^{*} = \va^{*} \vb^{*T}$, where $a^*_i = b^*_i = \sqrt{a_i b_i}$.

It is not just symmetric matrices which are fixed points of synaptic balancing: all normal matrices---including symmetric matrices as a special case---are equilibria of our rule, i.e., they satisfy the balance condition \eqref{eq:balance-condition}. This more general result is shown with \S\ref{sec:normal-matrices-balanced}.

In purely linear networks, arbitrary similarity transformations of the recurrent matrix---not just positive diagonal transformations---are task-preserving.
The positive diagonal similarity transformation \eqref{eq:task-preserving-transformation}, then, will generically not achieve the optimal sensitivity among task-equivalent \emph{linear} networks.
It would be interesting to explore how optimizing over similarity transforms involving all invertible matrices might perform relative to our purely diagonal similarity transformation.
However, we note that such an optimization will not generically lead to a biologically plausible local learning rule.

\section{Synaptic balancing predicts heterosynaptic plasticity}
\label{sec:heterosynaptic-plasticity-perturbation}

\begin{figure*}
\makebox[\textwidth][c]{
\includegraphics[width=6in]{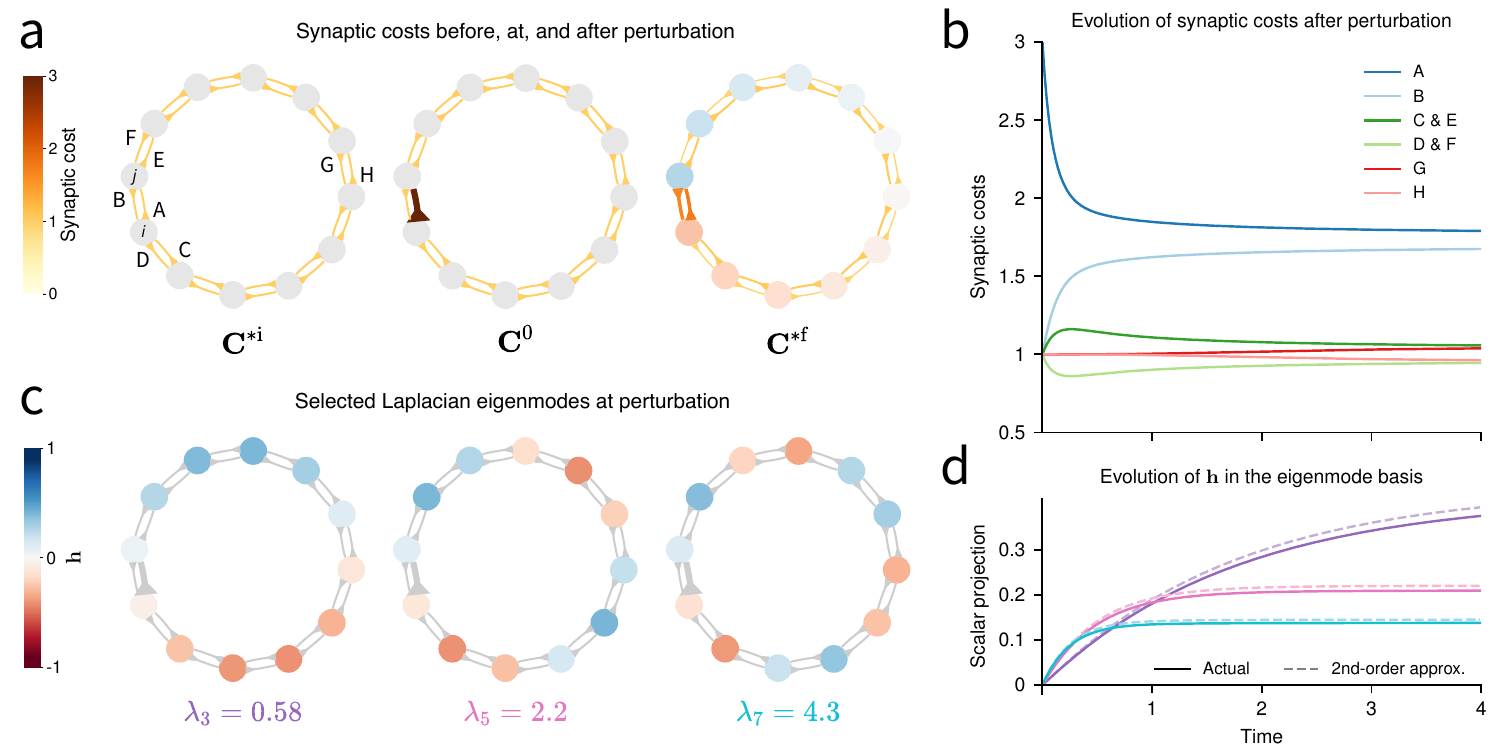}
}
\caption{
Dynamics of synaptic balancing in a 12-neuron ring network with single perturbed synapse.
\textbf{(a)}. 
Left: A ring network is at an initial equilibrium $\mC^\mathrm{*i}$ with all synaptic costs equal to $1$.
Nodes indicate neurons; arrows indicate directed synapses.
Center: Synapse A, from neuron $j$ to $i$, is instantaneously potentiated to a synaptic cost of $c_{ij}^0 = 3$.
Right: Synaptic balancing relaxes to a new equilibrium $\mC^\mathrm{*f}$.
Synapse colors and thickness indicate synaptic cost.  
Neuron colors indicate value of $\vh^\mathrm{*f}$ according to color scheme of panel (c).
\textbf{(b)}
Time course of synaptic balancing following perturbation.
Incoming synapses to neuron $i$ (A and D) are weakened and outgoing synapses from neuron $i$ (B and C) are strengthened.
Incoming synapses to neuron $j$ (B and E) are strengthened and outgoing synapses from neuron $j$ (A and F) are weakened.
Synapses (H and G) that are distant from the site of perturbation respond more slowly than proximate synapses though they reach the same equilibrium values.
\textbf{(c)}
Three eigenmodes $\vv_3, \vv_5, \vv_7$, with eigenvalues $\lambda_3, \lambda_5, \lambda_7$, of the Laplacian matrix $\mL^0$ corresponding to the conductance matrix at the moment of perturbation.
Color indicates mode value at each neuron.
\textbf{(d)}
Dynamics of $\vh$ approximately decompose into the basis of Laplacian eigenmodes. 
The scalar projection of $\vh$ onto each mode $\vv$ is shown along with the
quadratic approximation \eqref{eq:h-solution-equilibrium}, using $\mL^0$ as Laplacian.
Line color matches eigenvalue color in panel (c).
}
\label{fig:heterosynapic-plasticity-perturbation}
\end{figure*}

Previous sections demonstrated that the balance condition \eqref{eq:balance-condition} is not only functionally useful for noise robustness but also a sense generic in trained, regularized recurrent networks.
These findings suggest that synaptic balance may be a widespread phenomenon, with many networks attained through regularized training naturally exhibiting the equilibrium condition of synaptic balancing \eqref{eq:balance-condition}.
In a biological context, where we hypothesize that synaptic balancing would continually operate alongside other weight dynamics, we predict that the network is generically in a state fluctuating around the balance condition.
For example, fast-acting, inherently unstable plasticity rules---Hebbian or otherwise---might temporarily bring synapses and firing rates away from equilibrium, while synaptic balancing slowly tunes the network in response to those modifications, so that the balance condition is always approximately maintained.

This suggests that synaptic balancing might most commonly play a role in fine-tuning networks that are already in the vicinity of a balanced equilibrium.
To explore this scenario in a specific, experimentally testable regime, we consider the response of synaptic balancing to a single synaptic potentiation or depression near equilibrium.
\subsection{Synaptic perturbations induce compensatory heterosynaptic plasticity}
Suppose the that network begins at an initial equilibrium configuration $\mcW^\mathrm{*i}$ satisfying the balance condition \eqref{eq:balance-condition}, and that a perturbation is applied. Specifically, let the $i,j$th synapse ($i\neq j$) be instantaneously modified by a small factor $\eta \approx 0$ while every other synapse is unchanged.
Call the perturbed weight configuration $\mcW^0$. Then
\begin{align}
    \label{eq:Jij-perturbation}
    J_{kl}^0 = J^\mathrm{*i}_{kl} (1 + \eta \delta_{ik}\delta_{jl}).
\end{align}
If $\eta>0$, this perturbation corresponds to synaptic potentiation of $J_{ij}$, and if $\eta<0$, it corresponds to synaptic depression. 

By first-order expansion of the power-law synaptic costs \eqref{eq:cij-power-law}, we have that the costs adjust as
\begin{align}
    \label{eq:cij-perturbation}
    c^0_{kl} = c^\mathrm{*i}_{kl} (1 + \eta p \delta_{ik}\delta_{jl}) .
\end{align}
Plugging the perturbed synaptic costs \eqref{eq:cij-perturbation} into the expression for the neural gradient \eqref{eq:gk-costs}, and using that the network was initialized at equilibrium,  i.e. that $g^\mathrm{*i}_k = 0$ for all $k$, we have,
\begin{align}
    \label{eq:gk-perturbation}
    \begin{split}
        g^0_k
        &= \eta p c^\mathrm{*i}_{ij} (\delta_{ik} - \delta_{jk}).
    \end{split}
\end{align}
Finally, by plugging \eqref{eq:Jij-perturbation} and \eqref{eq:gk-perturbation} into \eqref{eq:Jij-dot}, the learning rule becomes, to first order,
\begin{align}
    \label{eq:Jkl-dot-perturbation}
    \dot J_{kl} = \eta p J_{kl}^\mathrm{*i} c_{ij}^\mathrm{*i}  (\delta_{il} - \delta_{jl} + \delta_{jk} - \delta_{ik})
\end{align}
This rule predicts that a perturbation from equilibrium at a single synapse from $j$ to $i$ will lead to multiplicative plasticity at the outgoing and incoming synapses of both neurons.
Specifically, if $ J_{ij}$ is potentiated, i.e., if $\eta > 0$, then the response of synaptic balancing is potentiation of the presynaptic neuron's incoming synapses ($J_{jl}$) and of the postsynaptic neurons' outgoing synapses ($J_{ki}$), as well as depression of the presynaptic neuron's outgoing synapses ($J_{kj}$) and the postsynaptic neuron's incoming synapses ($J_{il}$).

Following the perturbation from the initial equilibrium configuration $\mcW^\mathrm{*i}$, the system will relax under synaptic balancing to a final equilibrium configuration $\mcW^\mathrm{*f}$.
To calculate the change in synaptic strength of neuron $ij$ from its initial equilibrium to its final equilibrium value, we plug the value of $\vg^0$ from \eqref{eq:gk-perturbation} into \eqref{eq:hstar-solution-equilibrium} to find
\begin{align}
    \label{eq:hstar-solution-equilibrium-perturbation}
    \vh^\mathrm{*f}
    &= \eta c_{ij}^\mathrm{*i} \mL^\dagger (\ve_i - \ve_j),
\end{align}
where $\ve_k$ is the $k$th standard basis vector of $\reals^N$, and we adopt the Laplacian matrix corresponding to the perturbed weight configuration $\mcW^0$. 
Writing the task-preserving transformation $J_{ij}^\mathrm{*f} = J_{ij}^0 e^{h_j^\mathrm{*f} - h_i^\mathrm{*f}}$ in terms of
\eqref{eq:hstar-solution-equilibrium-perturbation},
then the log change in synaptic weight, normalized by the perturbation $\eta$, is
\begin{align}
    \begin{split}
    \label{eq:Jij-log-ratio-perturbation-equilibrium}
    \frac{1}{\eta} \log\left(\frac{J^\mathrm{*f}_{ij}}{J^0_{ij}} \right)
    &= - c_{ij}^0 ((L^\dagger)_{ii} + (L^\dagger)_{jj} - 2 (L^\dagger)_{ij})\\
    &= - c_{ij}^0 R_{ij}^0
    \end{split}
\end{align}
where $R_{ij}^0 = (L^\dagger)_{ii} + (L^\dagger)_{jj} - 2 (L^\dagger)_{ij} \geq 0$ is the {resistance distance} (alternately called effective resistance) between neurons $i$ and $j$, on the graph whose edges have conductances ${\bmC}^0$.
Resistance distance generalizes to arbitrary graphs the formulas for resistance in series and parallel, and $R_{ij}$ is greater where the paths of conductivity between neurons $i$ and $j$ are fewer or weaker
\cite{klein1993resistance}.
We may interpret \eqref{eq:Jij-log-ratio-perturbation-equilibrium} to mean that synaptic balancing counteracts a potentiation or depression of $J_{ij}$ by a factor equal to the fraction of overall conductance between neurons $j$ and $i$ that is attributable to the synaptic cost $c_{ij}^0$, versus to other paths of synaptic costs in the network.

\subsection{Slow, compensatory, and network-wide heterosynaptic plasticity}

We have described synaptic balancing near equilibrium as a slow, compensatory mechanism that adjusts a neuron's input and output synapses in response to a single potentiation or depression, in a manner closely resembling known phenomena of heterosynaptic plasticity \cite{chistiakova2014heterosynaptic, chistiakova2015homeostatic}.
Here we discuss experimental evidence in connection with the predictions of synaptic balancing.

Studies inducing Hebbian plasticity at one or more target synapses have repeatedly observed compensatory heterosynaptic modification at nearby synapses on the same dendrite \cite{dunwiddie1978long, royer2003conservation, oh2015heterosynaptic, el2018locally, field2020heterosynaptic}.
For example, two-photon \emph{in vivo} imaging of synaptic spines revealed that heterosynaptic depression of inputs to mouse V1 layer 2/3 pyramidal neurons follows functionally induced potentiation (LTP) of synapses on the same dendrite---a result exactly predicted by equation \eqref{eq:Jkl-dot-perturbation}.

Our specific predictions about the magnitude and direction of heterosynaptic effects also find experimental support. Equation \eqref{eq:Jkl-dot-perturbation} predicts heterosynaptic modifications that are multiplicative, with the change in weight proportional to the original size of the synapse, and independent of the synapse's sign (i.e., inhibitory or excitatory).
Patch-clamp recordings of synapses in cortical slice after induction of spike-timing-dependent Hebbian plasticity in paired synapses matched the multiplicative principle, with the heterosynaptic effect applying more strongly to the initially larger unpaired synapses than to the initially smaller ones \cite{field2020heterosynaptic}.
The same work found that heterosynaptic effects could be explained solely in terms of the absolute strengths of the paired and unpaired synapses: both E and I synapses experienced compensatory heterosynaptic depression when the paired synapse was potentiated, and both E and I unpaired synapses experienced compensatory potentiation when the paired synapse was depressed.

Our rule may be distinguished from existing concepts of compensatory heterosynapstic plasticity \cite{chistiakova2015homeostatic, zenke2017hebbian} through a further prediction, to our knowledge not yet experimentally explored: in response to the potentiation of a single neuron's inputs, 
not only do input synapses experience heterosynaptic depression, but also output synapses experience potentiation.
This concept is demonstrated in Fig. \ref{fig:heterosynapic-plasticity-perturbation} for a 12-neuron ring network that experiences an instantaneous synaptic perturbation.

Mechanisms supporting the input-synapse half of our conjectured balancing dynamics are well studied. In addition to the known heterosynaptic effects already mentioned, neurons are known to multiplicatively and bidirectionally scale all incoming synapses through synaptic scaling \cite{turrigiano1998activity, turrigiano2008self, hengen2013firing, turrigiano2017dialectic}.
While the particular homeostatic trigger posited by synaptic scaling---deviations from a set-point firing rate---differs from the trigger considered in our model, 
synaptic scaling lends evidence to the hypothesis that a neural mechanism exists to distribute a negative feedback signal to incoming synapses and induce coordinated compensatory plasticity, such as is predicted by synaptic balancing.

This work introduces a new view on the possible functional roles of compensatory heterosynaptic plasticity.
Existing literature has largely focused on the important role it may play in constraining the inherent instability of Hebbian plasticity
 \cite{chistiakova2014heterosynaptic, chistiakova2015homeostatic, zenke2017hebbian, Zenke2017-fo}.
In our model, the role of heterosynaptic plasticity is to maintain a functional state of noise robustness even as other learning processes homosynaptically modify synapse strength.
We believe that these traits are suitably thought of as a novel model of homeostatic plasticity: one whose aim is to maintain, through negative feedback, the functionally-relevant balance condition \eqref{eq:balance-condition}.

\section{Summary}
In this work, we introduced a positive diagonal similarity transformation of the recurrent weight matrix that preserves task performance in nonlinear recurrent networks with homogeneous nonlinearities. 
We showed that a simple class of cost functions, notably including the sensitivity of the neural dynamics to noise, are convex in the coordinates of the symmetry.
From this observation, we derived a local learning rule---synaptic balancing---that globally maximizes network robustness whenever the recurrent network is strongly connected.

We found that the synaptic cost matrix is balanced at equilibrium, and that this balance condition arises at every local minimum of a very general class of regularized loss functions.
To further understand how synaptic balancing dynamics shape network connectivity in the vicinity of equilibrium, we approximated the dynamics of our rule through a heat equation
and described the diffusion of synaptic modifications throughout the network according to its Laplacian eigenmodes.
We found that near equilibrium, synaptic balancing is well summarized as slow, compensatory, and heterosynaptic, and that experimental evidence of heterosynaptic plasticity is consistent with our predictions.

Overparameterization in neural network models may be linked to the biological processes which sustain task performance under noisy conditions---a possibility known to experimental neuroscience for some time \cite{marder2006variability}.
Here, we have provided a concrete, analytically tractable example of this concept, in which an identifiable symmetry in network parameterization gives rise to a corresponding local process for maintaining stable task performance.
We hope that this work may provide a fruitful framework for future research relating homeostatic processes to the mathematical structures underpinning neural network redundancy.

\section*{Acknowledgements}

The authors wish to thank Brandon Benson, Subhaneil Lahiri, and Jonathan Timcheck for helpful discussions and feedback.
C.H.S. thanks the Blavatnik Family Foundation. S.E.H. thanks the National Defense Science and Engineering Graduate Fellowship and the Stanford Graduate Fellowship.  S.G. thanks the James S. McDonnell and Simons Foundations, NTT Research, and an NSF CAREER Award for support. 

\printbibliography
\pagebreak

\appendix

\section{The TPT is task-preserving}\label{sec:prop1}

\textbf{Proposition 1.}\textit{ [Task-preserving transformation]
The transformation \eqref{eq:task-preserving-transformation} exactly preserves the input-output relationship of the neural dynamics \eqref{eq:neural-dynamics}.
Given two networks receiving the same time course of inputs $\vu(t)$ and with weight configurations $\tmcW$ and $\mcW = \pi_\vh(\tmcW)$ respectively, then:
\begin{enumerate}
    \item If $\tvx(t)$ is the time course of hidden unit neural activity under $\tmcW$, then $e^{-\mH} \tvx(t)$ is the time course of hidden unit neural activity under $\mcW$.
    \item If $\tvy(t)$ is the time course of output neural activity under $\tmcW$, then $\tvy(t)$ is also the time course of output neural activity under $\mcW$. 
\end{enumerate}}

\begin{proof}

From \eqref{eq:neural-dynamics}, the dynamics of the neural activity $\vx$ in the network with transformed weights $\mcW$ is
\begin{align*}
    \vf_\mcW(\vx, \vu)
    &= -\vx +  e^{-\mH}  \tmJ e^{\mH} \phi(\vx) + e^{-\mH}  \tWin \vu.
\end{align*}
Multiplying on the left by $e^\mH$, and using that $\phi$ is homogeneous, we have
\begin{align}
    \begin{split}
    \label{eq:rescaled-dynamics}
    e^{\mH}  \vf_\mcW(\vx, \vu)
    &=  -e^{\mH} \vx +  \tmJ \phi(e^{\mH} \vx) + \tWin \vu \\
    &=  \vf_{\tmcW}(\vz, \vu),
    \end{split}
\end{align}
where $\vz = e^\mH \vx$.
The time integral of the left hand side is the scaled neural activity of the network with weights $\mcW$:
\begin{align}
\begin{split}
    \label{eq:rescaled-dynamics-LHS-integral}
    \int_0^t e^{\mH}  \vf_\mcW(\vx, \vu) \dif t'
    &= e^{\mH} \int_0^t  \vf_\mcW(\vx, \vu) \dif t' \\
    &= \tau e^{\mH} \vx(t),
\end{split}
\end{align}
using that $\vx(0) = 0$.
The time integral of the right hand side of \eqref{eq:rescaled-dynamics} is the neural activity of the network with weights $\mcW^0$:
\begin{align}
\label{eq:rescaled-dynamics-RHS-integral}
    \int_0^t \vf_{\tmcW}(\vz, \vu) \dif t'
    &= \tau \vx^0(t).
\end{align}
Equating \eqref{eq:rescaled-dynamics-LHS-integral} and \eqref{eq:rescaled-dynamics-RHS-integral}, we find that
\begin{align}
\label{eq:tpt-scales-hiddens}
\vx(t) = e^{-\mH}\tvx(t),
\end{align}
demonstrating the first part of the proposition.

From \eqref{eq:readout} and \eqref{eq:tpt-scales-hiddens}, we have that the readout of the transformed network is
\begin{align*}
    \vy(t) &= \mWout \vx(t) \\
    &= (\tWout e^\mH) (e^{-\mH} \tvx(t)) \\
    &=  \tWout \tvx(t) \\
    &= \vy^0(t),
\end{align*}
demonstrating the second part of the proposition.
\end{proof}

\section{Existence and stability of equilibria}

\label{sec:equilibrium-existence-condition-appendix}

In this section we prove that under reasonable assumptions on the topology of the network, a minimum-cost weight configuration exists on the task-preserving manifold and is exponentially stable under synaptic balancing.
The results proved in this section imply Prop. \ref{prop:topology-convergence} of the main text.

We begin with some definitions related to network topology.
A connected component of the network refers to a connected component of the undirected graph whose edge weights are given by the (symmetric) conductance matrix $\bmC$ \eqref{eq:cij-bar-defn}.
Further, a set of neurons $\mcK$ is strongly connected if and only if for every $i,j \in \mcK$ one may follow a path of positive (directed) synaptic costs from $i$ to $j$.
Under the task-preserving transformation, synaptic costs which are initially zero remain zero, and costs which are initially positive remain positive, so topological properties of the network wiring diagram do not vary over the course of synaptic balancing.
Throughout the appendix we assume synaptic costs are of the power-law form \eqref{eq:cij-power-law}, unless explicitly stated otherwise.

\subsection{Total cost attains a global minimum on the task-preserving manifold in strongly connected networks}
We now establish a lemma on the boundedness of sublevel sets of the total cost $C$ in strongly connected networks.

As a preliminary comment, we generalize the observation made in the main text that $\sum_i h_i(t) = 0$ for all $t \geq 0$ \eqref{eq:sum-hi-zero}.
In particular, if $\mcK$ is a connected component of the network, then $\sum_{k \in \mcK} g_k = 0$
for every synaptic cost matrix.
 This may be seen by writing out $g_k$ in terms of synaptic costs \eqref{eq:gk-costs} and using that $c_{ij} =0$ if neurons $i$ and $j$ are in different connected components. 
 If a network has $K$ connected components and $\vk_1, \vk_2, \ldots, \vk_K \in \reals^N$ are indicator vectors for each component, then for each $i=1, 2, \ldots, k$,
 \begin{align}
    \label{eq:sum-hi-zero-connected-component}
\vk_i^T \vh = \vk_i^T \vh^0 = 0,
\end{align}
since $\vg = \dot \vh$.

\begin{lemma}
\label{lemma:state-space-bounded-strongly-connected}
Suppose that every connected component of a network is strongly connected, and that the network evolves under synaptic balancing from an initial total cost $C^0$.
Then there exists some compact box $\mcB \subset \reals^N$ which contains the sublevel set $\{\vh: C(\vh) \leq C^0\}$.
\end{lemma}

\begin{proof}
We give an intuitive upper bound on the individual synaptic costs, and we telescope it along paths of synapses in a strongly connected network to show that the sublevel set of $C^0$ is bounded.

To approximate the set of coordinates $\vh$ such that $C \leq C^0$,
we first note that a simple upper bound on the synaptic cost $c_{ij}$ for weight configurations in the sublevel set $\{\vh: C(\vh) \leq C^0\}$ is
\begin{align}
    \label{eq:cij-C0-inequality}
    c_{ij} \leq C^0,
\end{align}
since $c_{ij} \leq C$ and by assumption $C \leq C^0$.
Using that $c_{ij} = c_{ij}^0 e^{p(h_j - h_i)}$ \eqref{eq:cij-sensitivity-h}, we solve for $e^{h_j - h_i}$ in \eqref{eq:cij-C0-inequality}, assuming $c^0_{ij} \neq 0$:
\begin{align}
\begin{split}
    \label{eq:cij-bound}
    e^{h_j - h_i}
    &\leq (C^0/c^0_{ij})^{1/p}.
\end{split}
\end{align}

To obtain an analogous upper bound when $c_{ij} = 0$, now let $m$ and $n$ any be two neurons in the same connected component of the network.
By the strongly connected assumption, there is some path of neurons $(m, i_1, i_2, \ldots, i_\kappa, n)$ from $m$ to $n$ such that each consecutive directed synapse along the path has positive synaptic cost.
Telescoping \eqref{eq:cij-bound} along this path, we obtain
 \begin{align*}
    e^{h_m - h_n} 
    &= e^{h_m - h_{i_1}} e^{h_{i_1} - h_{i_2}} \cdots
     e^{h_{i_\kappa} - h_n} \\
     & \leq (C^0/c^0_{i_1 m})^{1/p} (C^0/c^0_{i_2 i_1})^{1/p} \cdots (C^0/c^0_{n i_\kappa})^{1/p}.
\end{align*}
Then apply this procedure along a path from neuron $n$ to neuron $m$ to obtain an upper bound for $e^{h_n-h_m}$.

The upper bounds for $e^{h_m-h_n}$ and $e^{h_n-h_m}$ imply that $|h_m - h_n|$ is bounded above for every $m,n$ in a shared a connected component. 
Combined with the constraint \eqref{eq:sum-hi-zero-connected-component}, it follows that $|h_m|$ is bounded above for all neurons $m$, and the sublevel set $\{\vh: C(\vh) \leq C^0\}$ is contained within some box $\mcB \subset \reals^N$.
\end{proof}

The previous lemma may be strengthened into a general result giving sufficient and necessary conditions on the existence of equilibria of synaptic balancing.
The following result is known \cite{hooi1970class, eaves1985line}; for completeness we state and prove it in the language of this paper.

\begin{proposition}
\label{prop:global-minimum-iff-strongly-connected}
The total cost $C$ \eqref{eq:total-cost} attains a global minimum on the task-preserving manifold
if and only if every connected component of the network is strongly connected.
\end{proposition}
\begin{proof}
In the first direction, assume that there is at least one pair of neurons $m,n$ in a connected component $\mcK$ such that $n$ cannot reach $m$ through a directed path of positive synaptic costs.
Our approach is to show that there is no weight configuration with this topology that can satisfy \eqref{eq:balance-condition}.

Suppose that a connected component of the network is partitioned into two sets of neurons $\mcI$ and $\mcJ$.
We sum \eqref{eq:balance-condition} over $k \in \mcI$ and use the fact that $c_{kl} = c_{lk} = 0$ if neurons $k$ and $l$ are in different connected components to obtain
\begin{align}
    \label{eq:balance-condition-aggregate}
    \sum_{i\in\mcI, j\in\mcJ} c_{ji}
    &=
    \sum_{i\in\mcI, j\in\mcJ} c_{ij}.
\end{align}
In short, if $C$ is minimized, then aggregate costs from $\mcI$ to $\mcJ$ are equal to the aggregate costs from $\mcJ$ to $\mcI$.

Let $\mcI\subset \mcK$ be the set of neurons which can reach $m$ through a directed path of synaptic costs, and
let $\mcJ\subset \mcK$ be set of neurons which cannot.
The sets $\mcI$ and $\mcJ$ are not empty; they include neurons $m$ and $n$ respectively.
On one hand, the synaptic cost $c_{ij}$ must be zero for every $i \in \mcI$ and $j \in \mcJ$ (else $j$ could reach $m$).
On the other hand, at least one cost $c_{ji}$, for some $i\in\mcI$ and $j\in\mcJ$, must be greater than zero (else $\mcI$ and $\mcJ$ would be in different connected components).
It is impossible for the total synaptic costs from $\mcI$ to $\mcJ$, which are positive, to equal the total synaptic costs from $\mcJ$ to $\mcI$, which are zero, and no solution exists to \eqref{eq:balance-condition-aggregate}---nor, by extension, to \eqref{eq:balance-condition}. 
As \eqref{eq:balance-condition} is satisfied by all global minima of $C$ on the task-preserving manifold, then no such minimum exists.

In the other direction, by Lemma \ref{lemma:state-space-bounded-strongly-connected}, there exists some compact box $\mcB \in \reals^N$ containing the sublevel set of $C^0$.
Let $C^* = \inf\{C(\vh): \vh \in \mcB\}$. 
As $C$ is a continuous function on the compact set $\mcB$, there is some value $\vh^* \in \mcB$ such that $C(\vh^*) = C^*$.
So $\vh^*$ globally minimizes $C$.
\end{proof}

\subsection{Global minima of the total cost are exponentially stable}

We have shown conditions under which an optimal weight configuration $\mcW^*$ exists on the task-preserving manifold.
We now argue, via an argument from strong convexity, that such a $\mcW^*$, when it exists, is uniquely determined and globally exponentially stable.
Strong convexity is the property that the curvature of the objective function has a positive lower bound, and it implies that a minimum is unique and exponentially stable under gradient descent \cite[Ch. 9]{boyd2004convex}.
Our approach is to re-parameterize the optimization problem and show that the cost function is strongly convex in the subspace of $\reals^N$ that is relevant to synaptic balancing dynamics.

\begin{proposition}
If a weight configuration $\mcW^*$ minimizes the total cost $C$ on the task-preserving manifold, then
(i) $\mcW^*$ is the unique weight configuration on the task-preserving manifold satisfying both \eqref{eq:balance-condition} and \eqref{eq:sum-hi-zero-connected-component}, and 
(ii) $\mcW^*$ is a globally exponentially stable equilibrium of synaptic balancing.
\end{proposition}
\begin{proof}
Claims (i) and (ii) both follow by showing that for every initial weight configuration $\mcW^0$, the cost function $C$ is strongly convex on a suitably defined re-parameterization of the task-preserving manifold, and that the dynamics of synaptic balancing is isomorphic to gradient descent dynamics in the strongly convex parameterization.

We begin with deriving a strongly convex formulation of the total cost.
The $\{ \vk_i \}_{i=1}^K$ in \eqref{eq:sum-hi-zero-connected-component} are mutually orthogonal, since they have non-overlapping nonzero entries, so
$\mathrm{dim}\{\vk_i \}_{i=1}^K$ = K.
Let $\mU$ be any $(N-K)\times K$ matrix whose columns form an orthonormal basis for the orthogonal complement of $\mathrm{span}\{\vk_i \}_{i=1}^K$.
Then the trajectory of synaptic balancing is contained within the range of $\mU$,
i.e., $\mU \mU^T \vh(t) = \vh(t)$, and
there exist reduced coordinates $\tilde \vh \in \reals^{N-K}$ satisfying
\begin{align}
    \label{eq:h-tilde-defn}
    \tilde \vh(t) = \mU^T \vh(t) \,
    \Longleftrightarrow
    \,
    \vh(t) = \mU \tilde \vh(t)
\end{align}
for all $t \geq 0$.
We abuse notation slightly, writing $C(\vh)$ and $C(\tilde \vh)$ to refer to the total cost when considered as a function of, respectively, the original and reduced coordinates.

Suppose that a network has initial cost matrix $\mC^0$.
By assumption, a global minimum exists at $\mcW^*$, so by Prop. \ref{prop:global-minimum-iff-strongly-connected}, the connected components of $\mC^0$ are each strongly connected.
Then by Lemma \ref{lemma:state-space-bounded-strongly-connected}, some compact box $\mcB \subset \reals^N$ exists such that the sublevel set $\{\vh \in \reals^{N}: \, C( \vh) \leq C^0 \}$ is contained in $\mcB$ for all $t\geq0$.
Then similarly $\{\tilde \vh \in \reals^{N-K}: \, C(\tilde \vh) \leq C^0 \} \subset \tilde \mcB$, where $\tilde \mcB$ is the projection of $\mcB$ onto the range of $\mU$.

We will now show that $C$ is strongly convex with respect to $\tilde \vh$ on $\tilde \mcB$.
This amounts to finding some $m>0$ such that for every $\tilde \vh \in \tilde \mcB$,
\begin{align}
    \label{eq:strong-convexity-condition}
    m
    &\leq
    \min_{\tilde \vu \in \reals^{N-K}:\, \|\tilde \vu\|=1}
    \tilde \vu^T \frac{\partial^2 C}{\partial \tilde \vh^2} \tilde \vu.
\end{align}
The Hessian of $C$ with respect to $\tilde \vh$ is, via \eqref{eq:h-tilde-defn} and \eqref{eq:C-hessian-laplacian}, 
\begin{align*}
    \frac{\partial^2 C}{\partial \tilde \vh^2}
    &= \mU^T \frac{\partial^2 C}{\partial \vh^2} \mU
    \nonumber
    \\
    &= p^2 \mU^T \mL \mU,
\end{align*}
where $\mL$ is the Laplacian matrix associated with the graph with weights $\bar \mC = \mC + \mC^T$.

A basic result on Laplacian matrices is that the null space of the Laplacian is the span of the indicator vectors of the connected components of the associated graph \cite[Ch. 1]{chung1997spectral}.
In our case, this means that $\mathrm{span} \{ \vk_i \}_{i=1}^K$ is the ($K$-dimensional) null space of $\mL(\tilde \vh)$ for all $\tilde \vh \in \reals^{N-K}$, since the task-preserving transformation does not alter network topology.

With these observations, the right hand side of \eqref{eq:strong-convexity-condition} reduces,
by the Courant-Fischer min-max theorem, to
\begin{align}
    \min_{\tilde \vu \in \reals^{N-K}:\, \|\tilde \vu\|=1}
    &
    \tilde \vu^T \frac{\partial^2 C}{\partial \tilde \vh^2} \tilde \vu
    \nonumber \\
    = p^2
    &
    \min_{\tilde \vu \in \reals^{N-K}:\, \|\tilde \vu\|=1}
    \tilde \vu^T \mU^T \mL \mU \tilde \vu
    \nonumber \\
    = p^2
    &
    \max_{U:\, \mathrm{dim}(U) = N-K} \, \min_{\vu \in U:\, \|\vu\|=1}
    \vu^T \mL \vu
    \nonumber \\
    \label{eq:lambda-K+1-L-derivation}
    = p^2
    &
    \lambda_{K+1},
\end{align}
and $\lambda_n$ is the $n$th eigenvalue of $\mL$, when ordered as $\lambda_1 \leq \lambda_2 \leq \ldots \leq \lambda_N$.

We have shown $\lambda_{K+1}(\mL(\tilde \vh))$ is positive for all values of $\tilde \vh$.
It remains to show that $\lambda_{K+1}$ has a positive lower bound for all $\tilde \vh \in \tilde \mcB$. This is immediate, however, from the fact that the eigenvalues $\{\lambda_i\}$, ordered by size as we have done here, are continuous functions of $\mL$ \cite[Ch. 1, \S3]{rellich1969perturbation}, and $\tilde \mcB$ is compact: then $\lambda_{K+1}$ must attain some minimum $\lambda_{K+1}^* > 0$ on $\tilde \mcB$.
Plugging \eqref{eq:lambda-K+1-L-derivation} into \eqref{eq:strong-convexity-condition}, we have
\begin{align*}
\min_{\tilde \vu \in \reals^{N-K}:\, \|\tilde \vu\|=1}
\tilde \vu^T \frac{\partial^2 C}{\partial \tilde \vh^2} \tilde \vu
&= p^2 \lambda_{K+1} 
\\
& \geq 
p^2 \lambda_{K+1}^*
\end{align*}
for all $\tilde \vh \in \tilde \mcB$.
So $C$ is strongly convex with respect to the reduced coordinates $\tilde \vh$ on the sublevel set of $C^0$.

Finally, we note that the linear relation \eqref{eq:h-tilde-defn} implies that the dynamics of gradient descent on $C$ with respect to $\vh$ is isomorphic to the dynamics gradient descent on $C$ with respect to $\tilde \vh$.
Thus, although $C$ is demonstrably \emph{not} strongly convex with respect to the standard task-preserving manifold coordinates $\vh \in \reals^N$, synaptic balancing dynamics nonetheless inherits the desirable properties of gradient descent on a strongly convex function.

Because strongly convex functions possess a unique minimum which is exponentially stable under gradient descent, every global minimizer $\mcW^*$ of $C$ on the task-preserving manifold is (i) the unique solution to \eqref{eq:balance-condition} and $\eqref{eq:sum-hi-zero-connected-component}$ on the task-preserving manifold and (ii) globally exponentially stable under synaptic balancing dynamics.
\end{proof}

\subsection{Normal matrices are balanced}
\label{sec:normal-matrices-balanced}

As discussed in \S\ref{sec:synaptic-balancing-dynamics}, maximizing the robustness is equivalent to minimizing the total cost, which can be written as a matrix Frobenius norm:

\begin{equation}\label{eq:normal_cost}
    C = \sum_{ij} \sigma_j^2 J_{ij}^2 = || \mJ \mathbf{\Sigma} ||_{F}^2
\end{equation}

where we have introduced the diagonal matrix of moments $\Sigma = \mathbf{diag \{\vsigma\}}$.
If $\mJ \mathbf{\Sigma}$ is a normal matrix, then the cost \eqref{eq:normal_cost} can be written in terms of its eigenvalues $\lambda_j$:

\begin{equation}
    || \mJ \mathbf{\Sigma} ||_{F}^2 = \sum_j |\lambda_j|^2
\end{equation}

Now consider an arbitrary diagonal similarity transformation of the weight matrix $\mJ$, resulting in a potentially non-normal matrix $\widetilde{\mJ \mathbf{\Sigma}} = \mQ^{-1} \mJ \mQ \mathbf{\Sigma}$.
Since the diagonal matrices $\mQ$ and $\mathbf{\Sigma}$ will commute, this transformation preserves the eigenvalues of $\mJ \mathbf{\Sigma}$.
If the matrix $\widetilde{\mJ \mathbf{\Sigma}}$ has Schur decomposition with unitary matrix $\mU$ and upper triangular matrix $\mathbf{\Lambda}$: 

\begin{equation}
    \widetilde{\mJ \mathbf{\Sigma}} = \mU \mathbf{\Lambda} \mU^{-1},
\end{equation}

then we can see that the Frobenius norm and thus the cost $C$ of this transformed matrix $\widetilde{\mJ \mathbf{\Sigma}}$ is equal to that of the matrix $\mathbf{\Lambda}$, which is always greater than or equal to the sum of the squared eigenvalues of the normal $\mJ \mathbf{\Sigma}$:

\begin{equation}
    ||\widetilde{\mJ \mathbf{\Sigma}} ||_{F}^2 = ||\mathbf{\Lambda} ||_{F}^2 \geq \sum_j |\lambda_j|^2
\end{equation}
since the eigenvalues of the matrix $\mJ \mathbf{\Sigma}$, $\widetilde{\mJ \mathbf{\Sigma}}$ and $\mathbf{\Lambda}$ are all shared.  
The inequality is saturated when the transformation $\mQ$ is unitary.
Therefore, we find that an arbitrary diagonal similarity transformation of the matrix $\mJ$ can not decrease the cost, in the case that $\mJ \mathbf{\Sigma}$ is normal.

Note that in the linear case, $\mathbf{\Sigma} = \mI$ and this argument holds for arbitrary invertible matrices $\mQ$ and normal weight matrices $\mJ$.

\section{Bounds on minimum value}
\label{sec:analytic-bounds}

In general, we are not aware of an exact solution to the minimum value $C^*$ of the total cost $C$ of synaptic balancing. 
In this section, we calculate upper and lower bounds on $C^*$ stated in Prop. \ref{prop:bounds}.
These bounds are computable based on the current state of the cost matrix, and so they help approximate the degree to which synaptic balancing will improve the robustness of a particular weight configuration.

\subsection{Lower bound}
\label{sec:geometric-mean-lower-bound-p}

We now state a basic bound on the synaptic costs that are attainable on the task-preserving manifold.

\begin{lemma}
\label{lemma:cij-geometric-mean-lower-bound}
A synaptic cost matrix $\mC$ evolving under synaptic balancing from an initial configuration $\mC^0$ satisfies, for all pairs of neurons $i,j$ and for all $t \geq 0$,
\begin{align}
    \label{eq:cij-bar-symmetric-lower-bound}
    \bar c_{ij}(t) \geq 2 \hat c_{ij},
\end{align}
where $\bar c_{ij}(t) = c_{ij}(t) + c_{ji}(t)$ \eqref{eq:cij-bar-defn} and $\hat c_{ij} = \sqrt{c_{ij}^0 c_{ji}^0}$ \eqref{eq:cij-geometric-mean}.
\end{lemma}

\begin{proof}
For any two neurons $i$ and $j$, the product of the reciprocal costs $c_{ij} c_{ji}$ is conserved by the task-preserving transformation.
This product-conserving property follows directly from writing out the product of power-law costs \eqref{eq:cij-power-law} in terms of the task-preserving transformation \eqref{eq:task-preserving-transformation}:
\begin{align}
\begin{split}
\label{eq:product-conserving-constraint}
  c_{ij} c_{ji}
  &= c_{ij}^0 c_{ji}^0 e^{h_j - h_i} e^{h_i - h_j}  \\
  &= c_{ij}^0 c_{ji}^0.
\end{split}
\end{align}

To derive a general lower bound on $\bar c_{ij}$, let $c^{*}_{ij}$ and $c^{*}_{ji}$ be the optimizers of the two-variable minimization problem incorporating the product-conserving constraint: 
\begin{align}
\begin{split}
\label{eq:Cprod-separated-problem}
\text{minimize} \quad & c_{ij} + c_{ji} \\
\text{subject to} \quad & c_{ij} c_{ji} = c^0_{ij} c^0_{ji} \\
& c_{ij}, c_{ji} \geq 0.
\end{split}
\end{align}
This, in turn, is equivalent to the single-variable problem
\begin{align}
\label{eq:Cprod-separated-problem-single-variable}
\begin{split}
\text{minimize} \quad & c_{ij} +   c^0_{ij} c^0_{ji} c_{ij}^{-1} \\
\text{subject to} \quad & c_{ij} \geq 0,
\end{split}
\end{align}
which is convex.
Analytically minimizing \eqref{eq:Cprod-separated-problem-single-variable}, we find that 
\eqref{eq:Cprod-separated-problem} has the unique solution
\begin{align*}
    \begin{split}
    c^{*}_{ij} = c^{*}_{ji} = \sqrt{c^0_{ij} c^0_{ji}} = \hat c_{ij}.
    \end{split}
\end{align*}
Thus, $\bar c_{ij} = c_{ij} + c_{ji} \geq c^*_{ij} + c^*_{ji} = 2 \hat c_{ij}$ for all $c_{ij}$, $c_{ji}$ accessible under the task-preserving transformation; in particular, the inequality holds along the trajectory of synaptic balancing.
\end{proof}

In short, the synaptic cost attains a minimum on the set of product-conserving matrices at the symmetric matrix whose elements are obtained by replacing each pair of initial reciprocal synaptic costs with its geometric mean.

As an additional consequence, we have that the geometric-mean matrix $\hat \mC$ is in fact the equilibrium cost matrix of synaptic balancing whenever $\hat \mC$ is accessible by the task-preserving transformation.

\begin{lemma}
\label{lemma:C-lower-bound-c-hat}
For all weight configurations on the task-preserving manifold, the total cost $C$ satisfies the lower bound 
\begin{align}
    \label{eq:symmetric-lower-bound-task-manifold}
    C \geq \sum_{ij} \hat c_{ij.}
\end{align}
If there is a weight configuration $\mcW^*$ on the task-preserving manifold whose cost matrix $\mC^*$ is symmetric, then:  $\mC^* = \hat \mC$;
equality is attained in \eqref{eq:symmetric-lower-bound-task-manifold} at $\mC^*$; and
$\mcW^*$ is a global minimizer of total cost on the task-preserving manifold.
\end{lemma}
\begin{proof}
The total cost is $C = \sum_{ij} c_{ij} = \frac{1}{2}\sum_{ij} \bar c_{ij}$,
and \eqref{eq:symmetric-lower-bound-task-manifold} is obtained by summing \eqref{eq:cij-bar-symmetric-lower-bound} over all synapses $i,j$.
If there a a symmetric cost matrix $\mC^*$ on the task-preserving manifold, then by the reciprocal-product-conserving constraint of the task-preserving transformation \eqref{eq:product-conserving-constraint}, it must be the geometric-mean matrix $\hat \mC$.
The remaining claims in the lemma follow immediately.
\end{proof}

Lemma \ref{lemma:C-lower-bound-c-hat} is incorporated in the main text as \eqref{eq:C-star-lower-bound-c-hat} in Prop. \ref{prop:bounds}.

\subsection{Upper bound}
Next we derive an upper bound on the optimal value of the total cost, computable based on the current state of the cost matrix.
Our approach is to bound the curvature of the sensitivity on the sublevel set of the initial total cost $C^0$ and then to minimize a quadratic upper envelope function using this bound.

Recall that $\frac{\partial^2 C}{\partial \vh ^2} = p^2 \mL$ \eqref{eq:C-hessian-laplacian}, i.e. the Hessian of the total cost takes the form of a Laplacian matrix whose elements are drawn from the conductance matrix $\bmC$. 
\eqref{eq:cij-bar-defn}.
Let $\lambda_\mathrm{max}$ be the maximum eigenvalue of the Hessian.
A simple upper bound is provided by Gershgorin's theorem:
\begin{align}
\lambda_\mathrm{max}
&\leq p^2 \max_i \sum_{j=1}^N |L_{ij}| \nonumber \\
&= p^2 \max_i \left( |L_{ii}| +  \sum_{j: j\neq i} |L_{ij}|  \right) \nonumber \\
&= p^2 \max_i \left( \sum_{j: j\neq i }  \bar c_{ij} +  \sum_{j: j\neq i }  \bar c_{ij} \right) \nonumber \\
\label{eq:C-bar-lambda-max-bound}
&= 2 p^2 \max_i \sum_{j: \, j \neq i}  \bar c_{ij}.
\end{align}

Since the maximum column sum of a matrix is upper bounded by the total sum of elements in the matrix, and in the case of $\bmC$ that total sum decreases during gradient descent, we have:
\begin{align}
    \max_i \sum_{j: j\neq i }  \bar c_{ij}
    &\leq \sum_{ij} \bar c_{ij} \nonumber \\
    &= 2 C \nonumber \\
\label{eq:C-bar-max-column-sum-bound}
    &\leq 2 C^0.
\end{align}
This bound is too large by a factor of roughly $N$ when synaptic costs are uniformly distributed throughout the network, but is (nearly) tight in the case of, for example, a hub-and-spoke network with a single central neuron projecting to all others, which each have a constant number of outgoing synapses.

Combining \eqref{eq:C-bar-lambda-max-bound} and \eqref{eq:C-bar-max-column-sum-bound}, we have in general that
$\lambda_\mathrm{max} \leq M$ where 
\begin{align}
\label{eq:lambda-max-upper-bound}
M 
&= 4p^2 C^0.
\end{align}

We use the bound on the maximum eigenvalue to obtain an upper bound on $C^*$, the optimal value of the total cost.
The fact that the curvature of $C$ as a function of $\vh$ is bounded implies that there exists a quadratic function $Q(\vh)$ tangent to $S$ at $\vh=0$,
\begin{align}
    \label{eq:upper-envelope-defn}
    Q(\vh) = C^0 + \vh^T \left. \frac{\partial C}{\partial \vh}\right|_{\vh=0}  + \frac{M}{2} \|\vh \|_2^2
\end{align}
such that $Q(\vh) \geq C(\vh)$ for every $\vh$ in the sublevel set of $C^0$ \cite[\S 9.1.2, eq. 9.13]{boyd2004convex}.
This is known as a quadratic upper envelope function, and
its minimum value $Q^*$ is
\begin{align}
\label{eq:upper-envelope-minimum}
Q^*
&= C^0 - \frac{1}{2M} \left\| \left.  \frac{\partial C}{\partial \vh} \right|_{\vh=0} \right\|_2^2 .
\end{align}
Via the gradient descent rule \eqref{eq:gradient-descent}, the initial gradient of the total cost is
\begin{align}
    \begin{split}
    \label{eq:initial-gradient}
    \left.  \frac{\partial C}{\partial h_k} \right|_{\vh=0}
    = -p g_k^0.
    \end{split}
\end{align}

By construction, the minimum of the upper envelope function upper bounds the minimum of $C$, i.e., $C^* \leq Q^*$.
Plugging in \eqref{eq:upper-envelope-minimum}, written in terms of  \eqref{eq:lambda-max-upper-bound} and \eqref{eq:initial-gradient}, we obtain the bound
\begin{align*}
\label{eq:C-star-lower-bound-supp}
    C^*
    &\leq Q^* \\
    &= C^0 - \frac{1}{8 C^0} \|\vg^0 \|_2^2,
\end{align*}
which is \eqref{eq:C-star-upper-bound} in the main text.

\section{Network training details}
\label{sec:network-training-details}

We constructed a context-dependent integration task with three input variables: the context $\va$ and two signals $\vs^{(1)}$ and $\vs^{(2)}$. Each variable was encoded as a pair of indicator-style inputs, yielding six dimensions of network inputs.
Each pair $(a_1, a_2)$, $(s^{(1)}_1, s^{(1)}_2)$, and $(s^{(2)}_1, s^{(2)}_2)$ independently took the value $(0,1)$ or $(1,0)$, plus Gaussian noise, for the duration ($T=50$) of each trial.
Three Boolean variables corresponded to eight total trial conditions.

The task target $\vz$ (also encoded as indicator variables $z_1, z_2$) was the integral of the noisy signal over time, as gated by the context:
\begin{align*}
    z_i(t) =
    \left\{
    \begin{array}{ll}
         \int_0^t s^{(1)}_i (t') dt', &  \va = (0, 1),\\
         \int_0^t s^{(2)}_i (t') dt', &  \va = (1, 0).
    \end{array}
    \right.
\end{align*}
for $i=1,2$.

Networks of the form \eqref{eq:neural-dynamics}, with rectified linear activation functions, were initialized with i.i.d. weights $J_{ij}\sim \mcN(0,N^{-1})$ and were trained via gradient descent to minimize an objective consisting of the squared-error loss function $\sum_{i,t} (y_i(t) - z_i(t))^2$, plus a regularization term $\lambda \sum_{ij} J_{ij}^2$.
Networks in Fig. \ref{fig:trained-networks}, with $N=256$ neurons, learned the task after being trained with a fixed learning rate of 0.003 for 1,600 training iterations. 

Average gains $\{\sigma^2_j\}_{j=1}^N$ were calculated across conditions and trials on a separate trial dataset not used in the training or evaluation of the network.
The balanced network was taken to be the final time step of a numerical solution to the ODE \eqref{eq:Jij-dot}, using neural gradients \eqref{eq:gk-costs} and synaptic costs \eqref{eq:cij-sensitivity}, with the original (trained) network as initial condition.
Fig. \ref{fig:trained-networks}b was attained by running the original and balanced networks with varying levels of additive Gaussian noise injected into the recurrent neural dynamics.

\end{document}